\documentclass[11pt]{article}
\usepackage{mathrsfs}
\usepackage{amsfonts}
\usepackage{amssymb}
\usepackage{bbm}
\usepackage{amsfonts,amssymb, mathrsfs, amsmath,amssymb,float}
\usepackage{amsthm}
\usepackage{amsmath}
\usepackage{algorithm,algorithmic}
\usepackage{blkarray}
\usepackage{color}

\usepackage{extarrows}

\usepackage{geometry}
\newtheorem{myDef}{Definition}[section]
\newtheorem{defn}[myDef]{Definition}
\newtheorem{prop}[myDef]{Proposition}
\newtheorem{exmp}[myDef]{Example}

\newtheorem{lem}[myDef]{Lemma}
\newtheorem{rem}[myDef]{Remark}
\newtheorem{thm}[myDef]{Theorem}
\newtheorem{cor}[myDef]{Corollary}
\newtheorem{alg}[myDef]{Algorithm}

\def\X{{\mathbb{X}}}
\def\Y{{\mathbb{Y}}}
\def\U{{\mathbb{U}}}
\def\V{{\mathbb{V}}}
\def\W{{\mathbb{W}}}
\def\G{{\mathbb{G}}}
\def\HS{{\mathbb{H}}}
\def\ES{{\mathbb{E}}}

\def\C{{\mathcal{C}}}

\def\deg{\hbox{\rm{deg}}}

\def\Z{{\mathbb{Z}}}

\def\N{{\mathbb{N}}}
\def\R{{\mathbb{R}}}
\def\C{{\mathbb{C}}}
\def\F{{\mathbb{F}}}
\def\E{{\mathbb{E}}}

\usepackage{environ}
\NewEnviron{iproof}[1][Proof]{\begin{proof}[\indent\bfseries #1]\BODY\end{proof}}{}

\def\FS{{\mathcal{F}}}
\def\GS{{\mathcal{G}}}
\def\IS{{\mathcal{I}}}
\def\CS{{\mathcal{C}}}
\def\BS{{\mathcal{B}}}

\def\bref#1{(\ref{#1})}

\def\bit{{{\rm{bit}}}}

\begin{document}

\title{\bf Quantum Algorithm for Optimization and Polynomial System Solving over Finite Field and Application to Cryptanalysis\thanks{The work is supported by grants NKRDPC No. 2018YFA0306702 and NSFC No. 11688101.}}
\author{Yu-Ao Chen$^{1,2}$, Xiao-Shan Gao$^{1,2}$, Chun-Ming Yuan$^{1,2}$\\
$^{1}$KLMM, Academy of Mathematics and Systems Science\\
 Chinese Academy of Sciences, Beijing 100190, China\\
$^{2}$University of Chinese Academy of Sciences, Beijing 100049, China\\
Email: xgao@mmrc.iss.ac.cn}
\date{}
\maketitle

\begin{abstract}
\noindent
In this paper, we give quantum algorithms for two fundamental
computation problems:
solving polynomial systems over finite fields
and optimization where the arguments of the objective
function and constraints take values from
a finite field or a bounded interval of integers.
The quantum algorithms can solve these problems
with any given success probability
and have polynomial runtime complexities in the size of the input,
the degree of the inequality constraints,  and
the condition number of certain matrices derived from the problem.
So, we achieved exponential speedup for these problems
when their condition numbers are small.
As applications, quantum algorithms are given to three
basic computational problems in cryptography:
the polynomial system with noise problem,
the short integer solution problem,
the shortest vector problem,
as well as the cryptanalysis for the lattice based NTRU cryptosystem.
It is shown that these problems and NTRU can against quantum computer attacks
only if their condition numbers are large, so
the condition number could be used as a new criterion
for the lattice based post-quantum cryptosystems.

%

\end{abstract}

\section{Introduction}
Solving polynomial systems and optimization over finite fields are fundamental
computation problems in mathematics and computer science, which are also  typical NP hard problems.
In this paper, we give quantum algorithms to these problems,
which could be exponential faster than the traditional methods under certain conditions.

\subsection{Main results}

Let $\F_q$ be a finite field, where $q=p^m$ for a prime number $p$ and $m\in\N_{\ge1}$.
Let $\FS=\{f_1,\ldots,f_r\}\subset\F_q[\X]$ be a set of polynomials in
variables $\X=\{x_1,\ldots,x_n\}$ and with {\em total sparseness} $T_\FS = \sum_{i=1}^r \#f_i$, where $\#f$ denotes the number of terms in $f$.
For $\epsilon\in(0,1)$, we show that
\begin{thm}\label{th-m1}
There is a quantum algorithm which decides whether $\FS=0$
has a solution in $\F_q^n$ and computes one if $\FS=0$ does have solutions in $\F_q^n$,
with success probability at least $1-\epsilon$
and complexity $\widetilde O(T_\FS^{3.5}D^{3.5}m^{5.5}
\log^{4.5} p\kappa^2\log1/\epsilon)$,
where $D={n+}\sum_{i=1}^n \lfloor\log_2 \max_{j}\deg_{x_i}f_j\rfloor$, $T_\FS$ is the total sparseness of $\FS$, and $\kappa$ is the condition number of $\FS$
(see Theorem \ref{th-qfq} for definition).
\end{thm}

The complexity of a quantum algorithm is the number of
quantum gates needed to solve the problem.
Since $T_\FS, D, \log p^m$ are smaller than the input size,
the complexity of the algorithm is polynomial
in the input size and the condition number,
which means that we can solve polynomial systems over finite fields using quantum computers with any given success probability and in polynomial-time
if the condition number $\kappa$ of $\FS$ is small, say when $\kappa$ is poly$(n,D)$.

We also give a quantum algorithm to solve the following optimization problem.
\begin{eqnarray}\label{eq-op}
\min_{\X\in\F_p^n,\Y\in\Z^m} o(\X,\Y)&&\hbox{ subject to }\cr
&&
  f_j(\X)=0 \mod p,\, j=1,\ldots,r; \\
&& 0\le g_i(\X,\Y)\le b_i,\, i=1,\ldots,s;\,0\le y_k\le u_k,k=1,\ldots,m, \nonumber
\end{eqnarray}
where $\FS=\{f_1,\ldots,f_r\}\subset\F_p[\X]$,
$\Y=\{y_1,\ldots,y_m\}$,
$\GS_o=\{o,g_1,\ldots,g_s\}\subset\Z[\X,\Y]$, and
$b_1,\ldots,b_s,$ $u_1,\ldots,u_m\in\N$.
The complexity of the algorithm is polynomial in
the size of the input, $\deg(g_i)$, $\deg(o)$,
and the condition number of the problem (see Theorem \ref{th-opt1} for definition).
Since Problem \ref{eq-op} is NP-hard, the algorithm gives an exponential speedup
over traditional methods if the condition number is small, say poly$(n,m)$.

Note that for $q=p$, Problem \bref{eq-op} includes polynomial system solving over
$\F_q$ as a special case. Problem \bref{eq-op} is meaningless for $\F_q$
with $q=p^m$ and $m>1$, since $\F_q$ cannot be embedded into $\Z$.

We apply our methods to three computational problems
widely used in cryptography:
the {\em polynomial systems with noise problem} (PSWN) \cite{alb1,JH,HL},
the {\em short integer solution problem} (SIS)  \cite{aj1},
the {\em shortest vector problem} (SVP) \cite{alb2,svp1,svp2}.
We also show how to recover the private keys
for the latticed based cryptosystem NTRU with our algorithm.
The complexity for solving all of these problems is polynomial
in the input size and their condition numbers.

The latticed based computational problems SVP and LWE are the bases
for 23 of the 69 submissions for the call by NIST
to standardize the post-quantum public-key encryption systems \cite{alb2}.
LWE is another important problem in cryptography and
can be reduced to the SIS problem~\cite{reg1}.
In theory, our results imply that the 23 proposed cryptosystems can against
the attack of quantum computers only if the related condition
numbers are large.
So, the condition number could be used as a new criterion
for lattice based post-quantum cryptosystems.

Let $p$ be a prime and $\FS=\{f_1,\ldots,f_r\}\subset\F_p[\X]$ with $r\gg n$.
The PSWN  is to find an $\X\in\F_p^n$  which satisfies the maximal number of
equations in $\FS$. The problem is also called MAX-POSSO \cite{alb1,HL}.
Our quantum algorithm for PSWN has complexity $\widetilde O(n^{3.5}T_\FS^{3.5}\log^{8} p\kappa^2)$, where $\kappa$ is the condition number of the problem.
The PSWN is very hard in the sense that,
even for the {\em linear system with noise} (LSWN) over $\F_p$,
to find an $\X$ satisfying more than $1/p$ of the equations is NP
hard \cite{JH,max-mq}.

Lattice-based cryptography began in 1996 with a seminal work by Ajtai \cite{aj1},
who presented a family of one-way functions based on the SIS.
%
The SIS problem is to find a solution  of a homogenous
linear system $A\X=0\mod p$ for $A\in\F_p^{r\times n}$,
such that $||\widehat{\X}||_2$  is smaller than a given bound.
Our quantum algorithm for SIS has complexity $\widetilde O((n\log p+r)^{2.5}(T_A\log p+n\log^2p)\kappa^2)$, where $T_A$ is the number of nonzero elements in $A$
and $\kappa$ is the condition number of the problem.

The SVP and CVP are two basic NP-hard problems widely used in cryptography.
The SVP is to find a nonzero vector with the smallest
Euclidean norm in a lattice in $\R^m$.
The CVP is to find a vector in a lattice, which
is closest to a given vector.
The SIS \cite{aj1} and LWE \cite{reg1} are the randomized versions of SVP and CVP, respcetively.
Our quantum algorithm for SVP has complexity
$\widetilde O(m(n^{7.5} + m^{2.5})(n^3+\log h)\log^{4.5}  h \kappa^2)$,
where  $n$ is the rank of the lattice,  $h$ is the maximal value in the generators of the lattice, and $\kappa$ is the condition number of  the problem.
Our quantum algorithm for CVP has a similar complexity.

NTRU is a  lattice-based public key cryptosystem
proposed  by Hoffstein, Pipher and Silverman \cite{NTRU},
which is  one of the most promising candidates for
post-quantum cryptosystems.
Our quantum algorithm can be used
to recover the private key from the public key in time
$\widetilde O(N^{4.5}\log^{4.5} q\kappa^2)$
for an NTRU with parameters $(N,p,q)$ with $q >p$.
In particular, we show that the three versions of NTRU recommended
in \cite{NTRU} have the desired security against quantum computers only if their condition numbers are large.

\subsection{Main idea of the algorithm}
\label{sec-i2}

Let $\FS\subset\C[\X]$ be a set of polynomials over $\C$.
A solution of $\FS$ is called Boolean if its components
are either $0$ or $1$. Similarly, a variable $x$ is called {\em a Boolean variable} if it satisfies $x^2-x=0$.
In \cite{qabes}, we give a quantum algorithm\footnote{No knowledge of quantum algorithm is needed to read this paper. What we do is to use traditional methods to reduce the problems
to be solved to this result.} to find
Boolean solutions of a polynomial system over $\C$,
which is called B-POSSO in the rest of
this paper.
The main idea of the quantum algorithms proposed in this paper
is to reduce the problem to be solved to B-POSSO,
under the condition that the number of variables and the
total sparseness of the new polynomial system is polynomial
in the size of the original polynomial system.

Our algorithm for  problem \bref{eq-op} consists of three main steps:
(1) The equational constraints $f_j(\X)=0 \mod p, j=1,\ldots,r$
are reduced into polynomial equations in Boolean variables over $\C$.
(2) The inequality constraints $0\le g_i(\X,\Y)\le b_i, i=1,\ldots,s$
are reduced into polynomial equations in Boolean variables over $\C$.
(3) The problem of finding the minimal value of the objective function
is  reduced several B-POSSOs.
We will give a brief introduction to each of these three steps below.

A key method used in our algorithm is to use a polynomial
in Boolean variables to represent the integers $0,1\ldots,b$ for
$b\in\Z_{>1}$.
Let $\theta_b(\G_\bit) = \sum_{k=0}^{\lfloor\log_2b\rfloor-1}G_k2^{k}+
(b+1-2^{\lfloor\log_2b\rfloor})G_{\lfloor\log_2b\rfloor}$,
where $\G_\bit =\{G_0,\ldots,G_{\lfloor\log_2b\rfloor}\}$
is a set of Boolean variables. Then, the values of
$\theta_b(\G_\bit)$ are exactly $0,1\ldots,b$.

For $\FS\subset\F_p[\X]$ and $\F_p = \{0,1,\ldots,p-1\}$,
we use three steps to reduce the problem of finding a solution of $\FS$ in $\F_p$
to a B-POSSO.
(1) $\FS$ is reduced to a quadratic polynomial system (MQ) $\FS_1$ by introducing new variables.
(2) Each variable in $\FS_1$ is expanded as
$x_i=\theta_{p-1}(\X_i)$ and $\FS_1$ is reduced to another
MQ $\FS_2$ in Boolean variables $\X_i=\{X_{ij},j=0,\ldots,\lfloor\log_2(p-1)\rfloor\}$.
Since $\FS_1$ is quadratic, the total sparseness of $\FS_2$ is well controlled.
(3) We obtain a polynomial over $\C$ from $\FS_2$ as follows
$\FS_3 = \{g-\theta_{\#g}(\U_g)p\,|\, g\in\FS_2\}$,
where $\U_g$ is a set of Boolean variables.
It is shown that solutions of $\FS$ in $\F_p$ can be recovered
from Boolean solutions of $\FS_3$, which can be found with the
quantum algorithm from \cite{qabes}.

We also reduce an inequality constraint
$0\le g(\X,\Y)\le b$ for $g\in\Z[\X,\Y]$ and $b\in\Z_{>0}$ into
a B-POSSO.
There exist $\X$ and $\Y$ such that $0\le g(\X,\Y)\le b$ if and only if
$g(\Y,\Y)-\theta_{b}(\G_g)=0$
has a solution for $\X$, $\Y$, and $\G_{g}$,
where $\G_g$ is a set of Boolean variables.
We reduce $g(\Y,\Y)$ into a polynomial in Boolean variables
by first reducing $g(\Y,\Y)$ into an MQ and then
expanding the variables into Boolean variables by using the $\theta$ function.
%
Let $d$ be the degree of $g_i$. Then the values of $g_i$
is exponential in $d$ and hence the number of Boolean variables needed
is polynomial in $d$. This is why the complexity of the algorithm depends
on $d$.

Since all variables are bounded, the objective function
$o$ is also bounded, and we can assume $o\in[0,u)$ for some $u\in\N$.
We design a novel search scheme to reduce
the minimization of $o(\Y,\Y)\in[0,u)$  to several B-POSSOs.
The minimal value of $o$ is found by bisecting the feasible interval $[0,u)$ recursively
into subintervals of the form $[\alpha,2^{\beta})$ and
deciding whether $o\in[\alpha,2^{\beta})$ has a solution,
which is equivalent to solving the equation $o-(\alpha+\sum_{j=0}^{\beta-1} {H_{j}} 2^j)=0$
for Boolean variables $H_{j}$.
As a consequence, we can find the minimal value of $o$ by solving
several B-POSSOs.

\subsection{Relation with existing work}

Problem \bref{eq-op} includes many important problems
as special cases, such as the
polynomial system solving over finite fields \cite{dingjt},
PSWN \cite{alb1,JH,HL,max-mq},
SIS  \cite{aj1},
SVP/CVP \cite{alb2,svp1,svp2},
the $(0,1)$-programming problem \cite{lip-survey},
the {\em quadratic unconstrained binary optimization problem} which is the mathematical problem that can be solved by the D-Wave System~\cite{Kadowaki},
which are all important computation problems and were widely studied.
%

Comparing to the existing work, our algorithm has two major advantages.
First, we give a universal approach to a very general problem.
Second, the complexity of our algorithm is polynomial in the inputs size,
the degree of the inequalities, and the condition number of
 the problem.
Since the problems under consideration are NP hard,
existing algorithms are exponential in some of
the parameters such as the number of variables.
In this aspect, we give a new way of looking at these NP hard problems
by reducing the computational difficulty to
the size of the condition number.

Our algorithm
is based on the quantum algorithm to solve B-POSSOs
proposed in \cite{qabes}, which in turn is based
on the HHL quantum algorithm and its variants to solve linear systems \cite{hhl,hhl-a,hhl-c}.
Comparing to the HHL algorithm, we can give the exact solution,
while the HHL algorithm can only give the quantum state.
The speedup of our algorithms comes from the HHL algorithm.
The limitation on the condition number is inherited
from the HHL algorithm, and it is proved in \cite{hhl} that the dependence on condition
number cannot be substantially improved.
Also note that, the best classic numerical method
for solving an order $N$ linear equation $Ax=b$ has complexity $\widetilde{O}(N\sqrt{\kappa})$, which also depends on the condition number  $\kappa$ of $A$ \cite{c19}.

%
%
%
%

The method of treating the inequality constraints with the
function $\theta_b(\G_\bit)$ simplifies the computational significantly.
The binary representation $\eta_b = \sum_{i=0}^{\lfloor\log_2(b)\rfloor}B_i2^i$ for $b$
is often used in the literature to represent the integers $0,1,\ldots,b$. The values of $\eta_b$
is $0,1,\ldots, 2^{\lfloor\log_2(b)\rfloor+1}-1$, which
may contain integers strictly larger than $b$ and cannot be used to represent
inequalities.
In \cite{arora1,fg1}, the inequality $0\le g\le b$ is reduced
to  $\prod_{i=0}^{b} (g-i)=0$.
Our reduction $\theta_b(\G_\bit)$ is better, which does not increase
the degree of the equation and the size of the equation
is increased in the logarithm scale,
while the method used in \cite{arora1,fg1} increases
the degree by a factor $b$ and increases the size
of the equation exponentially.
%

The rest of this paper is organized as follows.
In Section 2, we define the $\theta_b(\G_\bit)$ function and give an explicit formula to reduce
a polynomial system into an MQ.
In Section 3, we present the algorithm for solving polynomial systems  over finite fields.
In Section 4, we show how to reduce the inequality constraints in
problem \bref{eq-op} to a B-POSSO.
In Section 5, we present the algorithm for solving problem \bref{eq-op}.
In Section 6, we present a quantum algorithm for  PSWN.
In Section 7, we present a quantum algorithm for SIS.
In Section 8, we present a quantum algorithm for SVP/CVP.
In Section 9, we present a quantum algorithm to recover the private key for NTRU.
In Section 10, conclusions are given.

\section{Two basic reductions}
In this section, we give two basic reductions frequently used in the paper:
to represent an integer interval with a Boolean polynomial
and to reduce a polynomial system to an MQ.

\subsection{Represent an integer interval with a Boolean polynomial}
A variable $X$ is called a {\em Boolean variable} if it satisfies $X^2-X=0$.
In this paper, we use uppercase symbols to represent  Bollean variables.
A polynomial is called a {\em Boolean polynomial} if it is in a set of Boolean variables.
In this section, we will construct a Boolean polynomial whose values are
exactly the integers $0, 1,\ldots,b$ for a given positive integer $b>0$.

Set $s=\lfloor\log_2(b)\rfloor$ and
introduce $s+1$ Boolean variables $\mathbb{B}_\bit=\{B_0,\ldots,B_s\}$.
Inspired by $b=(2^s-1)+(b+1-2^s)$, we introduce the function $\theta_b(\mathbb{B}_\bit)$:
$\theta_1(\mathbb{B}_\bit)= G_0$ and for $b>1$
\begin{equation}\label{eq-theta}
\theta_b(\mathbb{B}_\bit)=\sum_{i=0}^{s-1}2^iB_i+(b+1-2^s)B_{s}.
\end{equation}
\begin{lem}\label{lm-bb}
When evaluated in $\C$ or $\F_p=\{0,1,\ldots,p-1\}$ with $p> b$, $\theta_b(\mathbb{B}_\bit)$ is a surjective map from $\{0,1\}^{s+1}$ to $\{0,1,\ldots,b\}$.
Furthermore,
$\#\mathbb{B}_\bit = \#\theta_b(\mathbb{B}_\bit)=\lfloor\log_2(b)\rfloor+1$.
\end{lem}
\begin{proof}
We first assume that $\theta_b(\mathbb{B}_\bit)$ is evaluated over $\C$.
It is easy to check this lemma when $b=1$.
When $b>1$, from the definition of $s$, we have $b/2<2^s\le b$ and hence $2^s-1< b$.
Since the values of $\sum_{i=0}^{s-1}2^iB_i$ are $0,1,\ldots,2^s-1$, for any integer $n\in[0,2^s-1]$, $n$ has a preimage of map $\theta_b(\mathbb{B}_\bit)$, where $B_{s}=0$.
Now consider an integer $n\in[2^s,b]$.
Since $n \ge 2^s$, we have  $n-(b+1-2^s) \ge 2 \cdot 2^s -b -1 > 2 \cdot b/2 +1 -b-1>-1\ge0$.
Since $n \le b$, we have  $n-(b+1-2^s) \le 2^s -1$.
Thus $0\le n-(b+1-2^s)\le 2^s-1$, and then $n$ has a preimage of map $\theta_b(\mathbb{B}_\bit)$, where $B_{s}=1$.
It is clear  $\#\mathbb{B}_\bit = \lfloor\log_2(b)\rfloor+1$.
Since $b+1-2^s>0$, we have
$\#\theta_b(\mathbb{B}_\bit)=\lfloor\log_2(b)\rfloor$+1.
The lemma is also valid when $\theta_b(\mathbb{B}_\bit)$ is evaluated over $\F_p$,
since all values in the computation are $\le p-1$.
\end{proof}

For instance,
$\theta_6(\mathbb{B}_\bit)=B_0+2B_1+3B_{2}$,
$\theta_7(\mathbb{B}_\bit)=B_0+2B_1+4B_{2}$,
$\theta_8(\mathbb{B}_\bit)=B_0+2B_1+4B_{2}+B_{3}$.
\begin{rem}
It is easy to check that $\theta_b$ is injective
if and only if $b= 2^k-1$ for some positive integer $k$.
For instance, $\theta_6$ is not injective: $3$ has two preimages
$B_0=1,B_1=1, B_2=0$ and $B_0=0,B_1=0,B_2=1$.
\end{rem}
%

\subsection{Reduce polynomial system to MQ}
It is well known that a polynomial system can be reduced to an MQ
by introducing some new indeterminates. In this section,
we give an explicit reduction which is needed in the complexity analysis in this paper.

For any field $F$, let $F[\X]$ be the polynomial ring over $F$ in the indeterminates $\X=\{x_1,\ldots,x_n\}$.
Denote the sparseness (number of terms) of $f\in F[\X]$ as $\#f$.
For $\FS=\{f_1,\ldots,f_r\}\subset F[\X]$,
denote $T_{\FS}=\sum_{i=1}^r\#f_i$ to be the {\em total sparseness} of $\FS$,
$N_{\FS}=\#\X=n$ to be the {\em number of indeterminates} in $\FS$,
$d_i=\max_j \deg_{x_i}(f_j)$ to be the degree of $\FS$ in $x_i$,
$M(\FS)$ to be the set of all monomials in $\FS$,
and $C(\FS)$ to be the size of the coefficients of the polynomials in $\FS$,
$(\FS)_{F[\X]}$ to be the ideal generated by $\FS$ in $F[\X]$.

We want to introduce some new indeterminates to rewrite $\FS$ as an MQ.

\begin{lem}\label{lm-Q}
Let $\FS=\{f_1,\ldots,f_r\}\subset F[\X]$. We can introduce a set of new indeterminates $\V$ and an {\rm MQ} $Q(\FS)\subset F[\X,\V]$ such that
$(\FS)_{F[\X]}= (Q(\FS))_{F[\X,\V]} \cap F[\X]$.
Furthermore, we have  $\#\V=(T_{\FS}+1)\sum_{i=1}^n\lfloor\log_2d_i\rfloor+nT_{\FS}= O(T_{\FS}D)$,
$N_{Q(\FS)}= n + \#\V =O(T_{\FS}D)$,
$\#Q(\FS)=r+\#\V=O(T_{\FS}D)$,
$T_{Q(\FS)}=T_\FS + 2\#\V =O(T_{\FS}D)$, and $C(Q(\FS)) = C(\FS)$,
where $D=n + \sum_{i=1}^n\lfloor\log_2d_i\rfloor$ and $d_i=\max_j \deg_{x_i}(f_j)$.
%
\end{lem}
\begin{proof}
If $\FS$ is already an MQ, set $Q(\FS) = \FS$ and $\V=\emptyset$.
First, we introduce new indeterminates $u_{ij}$ for $j=1,\ldots,\lfloor\log_2d_i\rfloor$ and new polynomials $u_{i1}-x_i^2$ and $u_{{i(j+1)}}-u_{ij}^2$ for $j=2,\ldots,\lfloor\log_2d_i\rfloor-1$.
It is clear that $x_i^{2^j}=u_{ij}$.
Without loss of generality, we assume $d_i\ge2$
and if $d_i\le 1$, then we do not need these $u_{ij}$.
Let  $\X^\alpha=\prod_{i=1}^nx_i^{\alpha_i}$ be a monomial of $\FS$, and $\alpha_i=\sum_{k=1}^{l_i}2^{\nu_{ik}}$ be the binary representation of $\alpha_i\le d_i$, where $l_i\le \lfloor\log_2 \alpha_i \rfloor+1\le \lfloor\log_2d_i\rfloor+1$ and $\nu_{i1}<\cdots<\nu_{il_i}$.
Thus
\begin{equation*}
\X^\alpha=\prod_{i=1}^n x_i^{\sum_{k=1}^{l_i}2^{\nu_{ik}}}=
\prod_{i=1}^n\prod_{k=1}^{{l_i}}x_i^{2^{\nu_{ik}}}
\equiv\prod_{i=1}^n\prod_{k=1}^{{l_i}}u_{i\nu_{ik}}.
\end{equation*}
Rewrite these $\{u_{ij}\}$ as $\{u_i\,|\,i=1,\ldots,L_\alpha\}$ and we have $\X^\alpha=\prod_{i=1}^{L_\alpha}u_i$, where $L_\alpha=\sum_kl_k \le \sum_{i=1}^n (\lfloor\log_2d_i\rfloor+1)\le \sum_{i=1}^n \lfloor\log_2d_i\rfloor +n$.
To rewrite this product as an MQ, we introduce new
indeterminates $\{v_1,\ldots,v_{L_\alpha-2}\}$
and quadratic polynomials $v_1-u_{1}u_{2}$, $v_{i}-v_{i-1}u_{i+1}$ for $i=2,\ldots,L_\alpha-2$ and $\X^\alpha=v_{L_\alpha-2}u_{L_\alpha}$.
Denote
$Q(\X^\alpha)=\{v_1-u_{1}u_{2}, v_{i}-v_{i-1}u_{i+1}, i=2,\ldots,L_\alpha-2\}$.
Finally, we obtain an MQ
\begin{eqnarray}\label{eq-Q}
Q(\FS)&=&\{
u_{i1}-x_i^2, u_{{i(k_i+1)}}-u_{ik_i}^2, i=1,\ldots,n, k_i=2,\ldots,\lfloor\log_2d_i\rfloor-1;\cr
&& \widehat{f}_j,j=1,\ldots,r\}
\cup_{\X^\alpha\in M(\FS)} Q(\X^\alpha)\subset F[\X,\V],
\end{eqnarray}
where $\V=\{u_{i},v_k\}$ and $\widehat{f}_i$ is obtained by replacing $\X^\alpha$ by $v_{L_\alpha-2}u_{L_\alpha}$ in $f_i$.
For convenience, we denote
\begin{eqnarray}\label{eq-Qf}
\widehat{Q}(f_j)&=&\widehat{f}_j,j=1,\ldots,r.
\end{eqnarray}

Let $\V=\{u_{i},v_k\}$ be the set of new indetermiantes.
It is clear that the number of these $u_{ij}$ is $ \sum_{i=1}^n\lfloor\log_2d_i\rfloor$.
To represent $\X^\alpha$, we need $\sum_{i=1}^n l_i-2\le \sum_{i=1}^n\lfloor\log_2d_i\rfloor+n$
new indeterminates $v_i$.
In total, we have $\#V\le \sum_{i=1}^n\lfloor\log_2d_i\rfloor+T_{\FS}(\sum_{i=1}^nl_i-2)
\le(T_{\FS}+1)\sum_{i=1}^n\lfloor\log_2d_i\rfloor+nT_{\FS} =O(T_\FS D)$.
Then, $N_{Q(\FS)}= \#\X + \#\V =O(T_{\FS}D)$, since $n \le D$.
$\#Q(\FS)=r+\#\V=O(T_{\FS}D)$, since $r \le D$.
$Q(\FS)$ contains $r$ polynomials $\widehat{Q}(f_j),j=1,\ldots,r$
and $\#V$ binomials.
Then $T_{Q(\FS)}=T_\FS + 2\#\V =O(T_{\FS}D)$.
Since we only introduce new coefficients $\pm 1$, we have $C(Q(\FS)) = C(\FS)$.
\end{proof}

\begin{exmp}\label{ex-Q}
Let $\FS=\{f_1=x_1^3x_2^5+2x_1^{7}x_2^5+3\}$.
%
We have $d_1=7$, $d_2=5$, and  $Q_1=\{u_{11}-x_1^2,u_{12}-u_{11}^2,u_{21}-x_2^2,u_{22}-u_{21}^2\}$.
Then $x_1^3x_2^5 = x_1 u_{11}x_2u_{22} = v_2u_{22}$,
$x_1^{7}x_2^5 = x_1 u_{11}u_{12}x_2u_{22} = v_5u_{22}$,
where
$Q_2=\{v_1-x_1u_{11}, v_2-x_2v_1,
       v_3-x_1u_{11},v_4-v_3u_{12},v_5-v_4x_2\}$.
Finally,  $Q(\FS)=Q_1\cup Q_2\cup \{v_{2}u_{22}+2v_{5}u_{22}+3\}$.
Note that the above representation is not optimal and we can use less
new variables to represent $f_1 = x_1v_{2}+2x_1v_2x_1^4 +3=x_1v_{2}+2x_1v_3' +3$,
where $v_3'=v_2u_{12}$.
\end{exmp}
%

\begin{rem}
As mentioned in Example \ref{ex-Q}, the representation for $Q(\FS)$ is not optimal.
The  binary decision diagram (BDD) \cite{bdd} can be used to give a better
representation for $Q(\FS)$ by using less variables $v_i$.
\end{rem}

\section{Polynomial system solving over finite fields}
\label{sec-posso}
%
Let $\FS=\{f_1,\ldots,f_r\}\subset\F_q[\X]$ be a finite set of polynomials over the finite field $\F_q$, $t_i=\#f_i$, and $T_\FS = \sum_{i=1}^r t_i$. In this section, we give a quantum algorithm to find a solution of $\FS$ in $\F_q^n$. Denote the solutions of $\FS$ in $\F_q^n$ by $\V_{\F_q}(\FS)$.
For a prime number $p$, we use the standard representation $\F_p=\{0,\ldots,p-1\}$.

\subsection{Reduce MQ over $\F_p$ to MQ in Boolean variables over $\C$}
\label{ss-fp}
Let   $\FS=\{f_1,\ldots,f_r\}\subset\F_p[\X]$ be an MQ,
$t_i=\#f_i$, and $T_\FS = \sum_{i=1}^r t_i$.
In this section, we will construct a set of Boolean polynomials over $\C$,
from which we can obtain $\V_{\F_p}(\FS)$.
%
The reduction  procedure consists of the following two steps.

{\bf Step 1}. We reduce $\FS$ to a set of polynomials in Boolean variables over $\F_p$.
If $p=2$, then the $x_i$ are already Boolean and we can skip this step. We thus assume $p>2$ and set
%
%
\begin{eqnarray}
x_i&=&   \theta_{p-1}(\mathbb{X}_i) =
  \sum_{j=0}^{\lfloor\log_2(p-1)\rfloor-1}X_{i,j}2^j+
(p-2^{\lfloor\log_2(p-1)\rfloor})X_{i,\lfloor\log_2(p-1)\rfloor},\cr
\X_{i}&=&\{X_{ij}, j=1,\ldots, \lfloor\log_2(p-1)\rfloor\},\label{eq-xbit}\\
\X_{\text{bit}}&=&\cup_{i=1}^n \X_i = \{X_{ij}\,|\,i=1,\ldots,n,\ j=0,\ldots,\lfloor\log_2(p-1)\rfloor\}.\nonumber
 \end{eqnarray}
where $\theta_{p-1}$ is defined in \bref{eq-theta} and
$X_{ij}$  are Boolean variables.
Let $f_i=\sum _{j=1}^{t_i}c_{i,j}\X^{\alpha_{ij}}$, where $\alpha_{ij}=(\alpha_{ij}(1),\ldots,\alpha_{ij}(n))\in\N^n$.
Substituting \bref{eq-xbit} into $\FS$, we have
\begin{eqnarray}\label{eq-s1}
f_{i\text{bit}}&=&\sum_{j=1}^{t_i}c_{i,j}
\prod_{k=1}^n( \theta_{p-1}(\mathbb{X}_i))^{\alpha_{ij}(k)} \in\F_p[\X_i],\\
B(\FS)&=&\{f_{1\text{bit}},\ldots,f_{r\text{bit}}\}\subset\F_p[\X_\text{bit}].\nonumber
\end{eqnarray}
%
For any set $S$, set $$H_S=\{x^2-x\,|\,x\in S\}.$$
We have
\begin{lem}\label{lm-st11}
There is a surjective morphism
 $\Pi_1:\V_{\F_p}(B(\FS),H_{\X_\bit}) \rightrightarrows \V_{\F_p}(\FS)$,
 where $\Pi_1(\X_{\bit})$ $ = (\theta_{p-1}(\mathbb{X}_1),$ $\ldots, \theta_{p-1}(\mathbb{X}_n))$.
Furthermore, $\#\X_\bit = O(n\log{p})$ and  the total sparseness of $B(\FS)$ is
$O(T_\FS \log^2 p)$.
\end{lem}
\begin{proof}
%
By Lemma \ref{lm-bb}, it is easy to check that $\Pi_1$ is surjective.
By Lemma \ref{lm-bb}, {$\#\theta_{p-1}(\mathbb{X}_i)=\lfloor\log_2(p-1)\rfloor+1$} and hence
$\#\X_\bit = O(n\log{p})$.
Since $\FS$ is an MQ, for any monomial $\X^{\alpha_{ij}}$ of $f_i$,
we have $|\alpha_{ij}| \le 2$ and
$\prod_{k=1}^n( \theta_{p-1}(\mathbb{X}_i))^{\alpha_{ij}(k)}$ in
\bref{eq-s1} has at most $O(\log^2 p)$ terms.
Therefore, the total sparseness of $f_{i\bit}$ is $O(\#f_i \log^2 p)$
and the total sparseness of $B(\FS)$ is $O(T_\FS \log^2 p)$.
\end{proof}
{\bf Step 2}.
We introduce new Boolean indeterminates $U_{i,j}$ and reduce each $f_{i\bit}$ into a Boolean polynomial over $\Z$.
Let $t_i' = \#f_{i\text{bit}}$ and let
\begin{eqnarray}
\U_i&=&\{U_{i,j},j=0,\ldots,\lfloor\log_2 t_i'\rfloor\},\cr
\U_\bit&=&\cup_{i=1}^r\U_i=\{U_{i,j}\,|\,i=1,\ldots,r,\ j=0,\ldots,\lfloor\log_2 t_i'\rfloor\},\label{eq-U}\\
\theta_{ t_i'}(\U_i)&=& \sum_{j=0}^{\lfloor\log_2 t_i'\rfloor-1}U_{i,j}2^j+( t_i'+1-2^{\lfloor\log_2 t_i'\rfloor})U_{i,\lfloor\log_2 t_i'\rfloor}\in\F_p[\U_i],\label{eq-thetaU}\\
P(f_{i\bit})&=&f_{i\bit}-p\theta_{ t_i'}(\U_i)\in\Z[\X_\bit,\U_i],\cr
P(\FS)&=&\{P(f_{i\bit})\,|\,i=1,\ldots,r\}\subset\Z[\X_\bit,\U_\bit],\label{eq-P}
\end{eqnarray}
and we have
\begin{lem}\label{lm-crct}
  There is a surjective morphism $\Pi_2:\V_\C(P(\FS),H_{\X_{\bit}},H_{\U_{\bit}})
 \rightrightarrows \V_{\F_p}(\FS)$,
 where\\ $\Pi_2(\X_\bit,\U_\bit)$ $ =\Pi_1(\X_\bit)=
 (\theta_{p-1}(\X_1),\ldots,\theta_{p-1}(\X_n))$.
\end{lem}
\begin{proof}
Let $(\check \X_\bit,\check \U_\bit)\in\V_\C(P(\FS),H_{\X_{\bit}},H_{\U_{\bit}})$. Then $(\check \X_\bit,\check \U_\bit)$ is a Boolean solution of $P(\FS)\subset\Z[\X_\bit,\U_\bit]$ and
\begin{equation*}
0=P(f_i)(\check \X_\bit,\check \U_\bit)=f_{i\bit}(\check \X_\bit)-p\theta_{\lfloor C_i/p\rfloor}(\check \U_\bit)\equiv f_{i\bit}(\check \X_\bit)\equiv f_{i}(\Pi_1(\check\X_{\bit}))\pmod p,
\end{equation*}
where the last equivalence comes from \bref{eq-s1},
and $\Pi_1(\check\X_{\bit})=(\theta_{p-1}(\check\X_1),\ldots,\theta_{p-1}(\check\X_n))$ is defined in Lemma \ref{lm-st11}. As a conclusion, $f_{i}(\Pi(\check\X_{\bit}))\equiv0\pmod p$, or $(\theta_{p-1}(\check\X_1),\ldots,\theta_{p-1}(\check\X_n))\in\V_{\F_p}(\FS)$.

We now prove that $\Pi_2$ is surjective.
By Lemma \ref{lm-st11}, $\V_{\F_p}(B(\FS),H_{\X_\bit}) \rightrightarrows \V_{\F_p}(\FS)$,
so it is enough to prove $\V_\C(P(\FS),H_{\X_{\bit}},H_{\U_{\bit}})
\rightrightarrows\V_{\F_p}(B(\FS),H_{\X_\bit})$.
Let $\check \X_\bit\in\V_{\F_p}(B(\FS),H_{\X_\bit})$
and $f_{i\text{bit}}=\sum_{j=1}^{t_i'}c_{i,j}'
\X_\text{bit}^{\beta_{ij}}\in\F_p[\X_{\bit}]$, where $c_{i,j}'\in\{0,\ldots,p-1\}\subset\Z$. Denote $C_i=\sum_{j=1}^{t_i'}c_{i,j}'\le (p-1)t_i'$. Then, $f_{i\text{bit}}(\check \X_\bit)\equiv0\pmod p$ if and only if $f_{i\text{bit}}(\check \X_\bit)=0,p,2p,\ldots,\text{ or }t_i' p$,
since $\lfloor C_i/p\rfloor p\le t_i'$.
By Lemma \ref{lm-bb}, there exist Boolean variables $\check\U_i=\{\check U_{i,j},j=0,\ldots,\lfloor\log_2 t_i'\rfloor\}$ such that $f_{i\bit}(\check \X_\bit)=p\theta_{ t_i'}(\check \U_i)$.
Hence $(\check \X_\bit,\check\U_i)$ is a preimage of $\check \X_\bit$ for the map $\V_\C(P(\FS),H_{\X_{\bit}},H_{\U_{\bit}})
 \rightrightarrows\V_{\F_p}(B(\FS),H_{\X_\bit})$.
Then,  the map  $\Pi_2$ is surjective.
\end{proof}

Since the map in \bref{eq-xbit} is not injective,
this map $\Pi_2$ is also not injective.

\begin{lem}\label{lm-pf}
The polynomial system $P(\FS)$ defined in \bref{eq-P} is of total sparseness $T_{P(\FS)}=\widetilde{O}(T_\FS\log^{2} p)$ and has
$N_{P(\FS)}=O(n\log p+\sum_{i=1}^r\log t_i+r\log\log p)$ indeterminates.
%
Furthermore, we can compute $P(\FS)$ from $\FS$ in $\widetilde{O}(T_\FS\log^2 p)$ binary operations.
\end{lem}
\begin{proof}
By Lemma \ref{lm-st11},  $B(\FS)$ is  of total sparseness  $O(T_\FS \log^2 p)$ and has $O(n\log p)$ indeterminates.
Since $\FS$ is an MQ, by the proof of Lemma \ref{lm-st11},
we have  $t_i'=\#f_{i,\bit} \le t_i\log^2 p$.
Then, the number of $U_{i,j}$ introduces in \bref{eq-U} is $\#\U_\bit=\sum_{i=1}^r\lfloor\log_2 t_i'\rfloor=O(\sum_{i=1}^r\log t_i')=O(\sum_{i=1}^r(\log t_i+\log(\log p)^2))=O(\sum_{i=1}^r\log t_i+r\log\log p)$.
Therefore, the total number of indeterminates is  $\#\X_\bit+\#\U_\bit=O(n\log p+\sum_{i=1}^r\log t_i+r\log\log p)$.

From \bref{eq-P}, the total sparseness of $P(\FS)$ is $T_{P(\FS)}
=T_{B(\FS)}+\sum_{i=1}^r \#\theta_{t_i'}(\U_i)  =T_{B(\FS)}+\#\U_\bit
=O(T_\FS\log^2 p+\sum_{i=1}^r\log t_i+r\log\log p)=\widetilde{O}(T_\FS\log^2 p)$, since $r\le T_\FS$ and $\sum_{i=1}^r\log t_i\le\sum_{i=1}^r t_i=T_\FS$.

To compute each $2^j\mod p$ costs $O(\log p)$ binary operations. Using the fast polynomial arithmetics \cite{mca}, to expand all the polynomials in $B(\FS)$ costs $\widetilde{O}(T_\FS\log^2 p)$ binary operations.
The cost of other steps to obtain $P(\FS)$ is negligible.
\end{proof}

%
\begin{cor}\label{cor-line}
If $\FS$ is a linear system, then
$T_{P(\FS)}=O(T_\FS\log p)$ and $N_{P(\FS)}=\widetilde{O}(n\log p+\sum_{i=1}^r\log t_i+r\log\log p)$.
\end{cor}
\begin{proof}
Since each $f_i$ is linear, we have $T_{B(\FS)}=O(T_\FS\log p)$, and $T_{P(\FS)}=O(T_\FS\log p+\#\U_\bit)=\widetilde{O}(T_\FS\log p)$.
\end{proof}

\begin{rem}\label{rem-3.1}
In \bref{eq-thetaU},  we can use $\theta_{\lfloor C_i/p\rfloor}$  instead of $\theta_{ t_i'}$ to introduce less indeterminates.
To compute each $C_i=\sum_{j=1}^{t_i'} c_{ij}'$ costs  $t_i'\log p=O(t_i\log^3 p)$, and   to compute all $C_i$ costs   $O(T_\FS\log^3 p)$, which is more than $T_{P(\FS)}=O(T_\FS\log^2 p)$. But, this is negligible comparing to
the final complexity of the algorithm in Corollary \ref{cor-fsc}.
\end{rem}

\subsection{Solving polynomial systems over $\F_p$}
Let $\FS=\{f_1,\ldots,f_r\}\subset \F_p[\X]$.
By Lemma \ref{lm-Q}, we can convert $\FS$ into an MQ $Q(\FS)\subset\F_p[\X,\V]$.
By Lemma \ref{lm-crct}, we can convert $Q(\FS)$ to an MQ
in Boolean variables over $\C$: $P(Q(\FS))\subset\C[\X_\bit,\V_\bit,\U_\bit]$.
To solve $P(Q(\FS))$, we need the following result,
where a quantum algorithm for B-POSSO is given.
A solution $\mathbf a$ of $\BS\subset\C[\X]$ is called
a {\em Boolean solution} if each coordinate of $\mathbf a$ is either $0$ or $1$.
\begin{thm}[\cite{qabes}]\label{th-m2}
For a finite set $\BS\subset\C[\X]$ and $\epsilon\in(0,1)$,
there exists a quantum algorithm {\rm\bf QBoolSol} which decides whether $\BS=0$ has a Boolean solution and computes one if $\BS=0$ does have Boolean solutions, with probability at least $1-\epsilon$
and complexity
$\widetilde O(n^{2.5}(n+T_\BS)\kappa^2\log1/\epsilon)$,
where  $T_\BS$ the total sparseness of $\BS$ and
$\kappa$ is the condition number of $\BS$.
\end{thm}

Here is the main result of this section.
\begin{thm}\label{th-alp}
For $\FS=\{f_1,\ldots,f_r\}\subset \F_p[\X]$ and $\epsilon\in(0,1)$,
there exists a quantum algorithm to find a solution
of $\FS$ in $\F_p$ with probability at least $1-\epsilon$
and the  complexity of the algorithm  is
$\widetilde O(T_\FS^{3.5}D^{3.5}\log^{4.5} p\kappa^2\log1/\epsilon)$,
where $T_\FS=\sum_{i=1}^r\#f_i$ is the total sparseness of $\FS$, $D=n+\sum_{i=1}^n \max_j\lfloor\log_2(\deg_{x_i}(f_j))\rfloor$, and $\kappa$ is the condition number of $P(Q(\FS))$, also called the condition number of $\FS$.
\end{thm}

We first estimate the total sparseness of $P(Q(\FS))$.
\begin{lem}\label{lm-fsb}
$P(Q(\FS))$ is of total sparseness $O(T_\FS D \log^2 p)=O(nT_\FS \log d \log^2 p)$
and has $O(T_\FS$ $ D\log p)=O(nT_\FS \log d \log p)$ indeterminates,
where $D=n+\sum_{i=1}^n \max_j\lfloor\log_2(\deg_{x_i}(f_j))\rfloor$
and $d=\max\{2, \log_2(\deg_{x_i}(f_j)),i=1,\ldots,n,j=1,\ldots,r\}$.
\end{lem}
\begin{proof}
By {Lemma \ref{lm-Q}}, $N_{Q(\FS)}=O(T_\FS D)$, $T_{Q(\FS)} =O(T_\FS D)$,
and $\#{Q(\FS)} =O(T_\FS D)$.
By {Lemma \ref{lm-pf}}, $P(Q(\FS))$ is of total sparseness $O(T_\FS D\log^2 p)$.
By {Lemma \ref{lm-pf}}, $P(Q(\FS))$  has
$N_{P(Q(\FS))}=O(N_{Q(\FS)} \log p+ \sum_{f\in Q(\FS)} \log (\#f)+ \#Q(\FS) \log\log p)$ indeterminates.
From the proof of Lemma \ref{lm-Q}, $Q(\FS)$ contains $\widehat{f}_j,j=1,\ldots,r$ and
$\# V$ binomials.
Then, $\sum_{f\in Q(\FS)} \log (\#f)
=\sum_{j=1}^r \log (\#\widehat{f}_j) + \# V\log 2
= O(\sum_{j=1}^r \log t_j + T_\FS D) = \widetilde{O}(T_\FS D)$.
Then $N_{P(Q(\FS))} = \widetilde{O}(T_\FS D \log p+T_\FS D
+T_\FS D\log\log p)=\widetilde{O}(T_\FS D\log p)$.
Since $D=n+\sum_{i=1}^n \max_j\lfloor\log_2(\deg_{x_i}(f_j))\rfloor =O(n\log d)$,
we obtain the bounds involving $n$ and $d$.
\end{proof}

{\em Proof of Theorem \ref{th-alp}}.
We can find a solution of $\FS$ as follows.
Construct $P(Q(\FS))\subset\C[\X_\bit,\V_\bit,$ $\U_\bit]$
according to  Lemma \ref{lm-Q} and Lemma \ref{lm-crct}.
Let $\mathbf b =$ {\bf QBoolSol}$(P(Q(\FS)),\epsilon)$.
If $\mathbf b =\emptyset$ then the algorithm fails to find a solutioon.
Let $\mathbf b =(\check\X_\bit,\check\V_\bit,\check\U_\bit)$
and $\check \X_\bit=(\check{X}_{1,0},$ $\ldots,\check{X}_{1,\lfloor\log_2(p-1)\rfloor)},
\check{X}_{2,0},$ $\ldots,\check{X}_{n,\lfloor\log_2(p-1)\rfloor)})$.
Let
$\widehat{X}=(\theta_{p-1}(\check{\X}_{1}),\ldots, \theta_{p-1}(\check{\X}_{n}))$,
where $\check{\X}_{i} = (\check{X}_{i,0},\ldots,\check{X}_{i,\lfloor\log_2(p-1)\rfloor})$.
By  Lemma  \ref{lm-Q}, Lemma \ref{lm-crct}, and Theorem \ref{th-m2},
$\widehat{X}$ is a solution of $\FS$ in $\F_p$ with probability at least $1-\epsilon$.

We now give the complexity.
By {Lemma \ref{lm-fsb}}, $P(Q(\FS))$ is of sparseness $O(T_\FS D\log^2 p)$ and has $O(T_\FS D\log p)$ indeterminates. By {Theorem \ref{th-m2}}, we can find a Boolean solution of $P(Q(\FS))$ in time
$\widetilde O((T_\FS D\log p)^{2.5}(T_\FS D\log p+T_\FS D\log^2 p)\kappa^2\log1/\epsilon)
=\widetilde O(T_\FS^{3.5}D^{3.5}\log^{4.5} p\kappa^2\log1/\epsilon).$
The complexity for other steps can be neglected.
\qed


Let $d=\max\{2, \max_{i=1}^n \max_j\deg_{x_i}(f_j)\}$.
Then $D=O(n\log d)$.
Since the solutions are in $\F_p$, we can assume $d < p$.
By  Theorem \ref{th-alp}, we have
\begin{cor}\label{cor-fsc}
The complexity to find a solution for $\FS=0\mod p$   is $\widetilde O(n^{3.5}T_\FS^{3.5}\log^{3.5} d \log^{4.5}p$ $\kappa^2\log1/\epsilon)$
=
$\widetilde O(n^{3.5}T_\FS^{3.5}\log^{8} p\kappa^2$ $\log1/\epsilon)$.
\end{cor}

\begin{cor}\label{cor-mqfp}
If $\FS$ is an {\rm MQ}, then the complexity is
$\widetilde O((n\log p+r)^{2.5}T_\FS\log^{2} p\kappa^2\log1/\epsilon)$.
\end{cor}
\begin{proof}
If $\FS$ is an MQ, we have $P(Q(\FS))=P(\FS)$.
By {Lemma \ref{lm-pf}}, $N_{P(\FS)}= O(n\log p+\sum_{i=1}^r\log t_i+r\log\log p)$ and $T_{P(\FS)}=O(T_\FS\log^{2} p)>N_{P(\FS)}$.
Since $\sum_{i=1}^r\log t_i =\log (\prod_{i=1}^r  t_i) < (\log \sum_{i=1}^r t_i)^r
=(\log T_\FS)^r$,
by {Theorem \ref{th-m2}}, the complexity is $\widetilde O(
N_{P(\FS)}^{2.5}T_{P(\FS)}\kappa^2\log1/\epsilon)  =\widetilde O((n\log p+\sum_{i=1}^r\log t_i+r\log\log p)^{2.5}T_\FS\log^{2} p\kappa^2\log1/\epsilon)=\widetilde O((n\log p+r\log(T_\FS/r))^{2.5}T_\FS\log^{2} p\kappa^2\log1/\epsilon)=\widetilde O((n\log p+r)^{2.5}T_\FS\log^{2} p\kappa^2\log1/\epsilon)$.
In the last step, we here use the reduction $(a+b\log c)(c+d)\le(a+b)\log (c+d)(c+d)=\widetilde O((a+b)(c+d))$.
\end{proof}
\begin{cor}
If $p=2$, then the complexity to find a solution of $\FS=0\mod 2$ is $\widetilde O((n+r)^{2.5}(n+T_\FS)\kappa^2\log1/\epsilon)$.
\end{cor}
\begin{proof}
If $p=2$, then we do not need to convert $\GS$ to MQ and $T_{P(\FS)}=O(T_\FS+\sum_{i=1}^r\log t_i)=\widetilde{O}(T_\FS)$,  $N_{P(\FS)}=O(n+\sum_{i=1}^r\log t_i)$.
Similar to the proof of Corollary \ref{cor-mqfp},  the complexity is $\widetilde O((n+\sum_{i=1}^r\log t_i)^{2.5}(n+T_\FS)\kappa^2\log1/\epsilon)=\widetilde O((n+r)^{2.5}(n+T_\FS)\kappa^2\log1/\epsilon)$.
\end{proof}


\subsection{Polynomial equation solving over $\F_q$}
In this section, we consider polynomial equation solving in
a general finite field $\F_q$ by reducing the problem
to equation solving over $\F_p$.

If $q=p^m$ with $p$ a prime number and $m\in\Z_{>1}$, then $\F_q=\F_p(\theta)$, where  $\varphi(\theta)=0$ for a monic irreducible polynomial $\varphi$ with $\deg(\varphi)=m$.
Let $g\in\F_q[\X]=\F_p[\theta,\X]$.
By setting $x_i=\sum_{j=0}^{m-1} x_{ij}\theta^j$  and write
each coefficient $c =\sum_{j=0}^{m-1} c_{j}\theta^j$ in $g$,
$g$ can be written as $g = \sum_{j=0}^{m-1} g_j \theta^j$, where
$g_j\in\F_p[\X_{\theta}]$ and
$\X_{\theta}=\{x_{ij}\,|\, i=1,\ldots,n,j=0,\ldots,m-1\}$
are variables over $\F_p$.
 We denote
$G(g) = \{g_0, g_1, \ldots,g_{m-1}\}\subset\F_p[\X_{\theta}].$
For a polynomial set $\FS\subset\F_q[\X]$, we denote
\begin{equation}\label{eq-G}
G(\FS) =\bigcup_{f\in\FS} G(f)\subset\F_p[\X_{\theta}].
\end{equation}

\begin{lem}\label{lm-mqq}
There is an isomorphism $\Pi_q:\V_{\F_p}(G(\FS))
\rightarrow \V_{\F_q}(\FS)$,
where $\Pi_q(x_{ij}) = (\sum_{j=0}^{m-1}x_{1j}\theta^j,$ $\ldots,
  \sum_{j=0}^{m-1}x_{nj}\theta^j)$.
Furthermore, for an {\rm MQ} $\FS=\{f_1,\ldots,f_r\}\subset\F_q[\X]$ with total sparseness $T_\FS$,  $G(\FS)\subset\F_p[\X_\theta]$ is an {\rm MQ} with total sparseness $\le m^3T_\FS$, $\#G(\FS)=mr$, and $\#\X_\theta=mn$.
\end{lem}
\begin{proof}
It is easy to {show} that $\#G(\FS)=m\#\FS$, $\#\X_\theta=m\#\X$ and $G(\FS)$ is also an MQ. Then the total sparseness of $G(\FS)$ will be concerned. $\FS$ has $T_\FS$ terms, where each term is of degree $\le2$.
For $x=\sum_{i=0}^{m-1} x_i\theta^i$ and $y=\sum_{i=0}^{m-1} y_i\theta^i$, let $c\theta^k =\sum_{j=0}^{m-1} c_{jk}\theta^j\mod\varphi(\theta)$ for any $k\in\N$, then we have $cxy=\sum_{i=0}^{m-1} \sum_{j=0}^{m-1} x_iy_{j}\sum_{k=0}^{m-1} c_{k\,(i+j)}\theta^{k}=\sum_{i=0}^{m-1} g_i\theta^i$, where $g_i\in\F_p[x_0,\ldots,x_{m-1},y_0,\ldots,y_{m-1}]$ is a quadratic polynomial with $T_{G(cxy)}\le\sum_{i=0}^{m-1} \sum_{j=0}^{m-1} \sum_{k=0}^{m-1} 1=m^3$.
Thus an MQ $\FS$ over $\F_q$ can be represented as another MQ $G(\FS)$ over $\F_p$ with $T_{G(\FS)}\le m^3T_\FS$.
\end{proof}
%
%
%
%
%
%
%
%
%
%
We have
\begin{thm}\label{th-qfq}
There is a quantum algorithm to find a solution of
$\FS\subset\F_q[\X]$ with probability at least $1-\epsilon$
and in time
$\widetilde O(m^{5.5}T_\FS^{3.5}D^{3.5}\log^{4.5} p\kappa^2\log1/\epsilon)$,
where $T_\FS$ is the total sparseness of $\FS$, $D=n+\sum_{i=1}^n \lfloor\log_2 \max_{j}\deg_{x_i}f_j\rfloor$, and $\kappa$ is the condition number of $P(G(Q(\FS)))$.
\end{thm}
\begin{proof}
Using Lemma \ref{lm-mqq}, we can solve $\FS$ over $\F_q$
similar to the method given in the proof of Theorem \ref{th-alp}.
In stead of solving $\FS_1=P(Q(\FS))\subset\C[\X_{\bit},\V_{\theta\bit},\U_{\theta\bit}]$
with algorithm {\bf QBoolSol},
we now solve $\FS_1=P(G(Q(\FS)))\subset\C[\X_{\theta\bit},\V_{\theta\bit},\U_{\theta\bit}]$
with algorithm {\bf QBoolSol},
where $Q(\FS)$ is defined in \bref{eq-G} and
$\X_{\theta\bit}$ is the bit representation for
$\X_\theta$.
%

We now prove the complexity.
By Lemma \ref{lm-Q}, $Q(\FS)$ has $O(T_\FS D)$ indeterminates and total sparseness $O(T_\FS D)$. $G(Q(\FS))$ has $O(mT_\FS D)$ indeterminates and total sparseness $O(m^3T_\FS D)$.
Since $G(Q(\FS))$ is an MQ, by {Corollary \ref{cor-mqfp}}, the complexity is
$\widetilde O(((mT_\FS D\log p+m(r+T_\FS D))^{2.5}(m^3T_\FS D)$ $\log^{2} p\kappa^2\log1/\epsilon)=\widetilde O(m^{5.5}T_\FS^{3.5}D^{3.5}\log^{4.5} p\kappa^2$
 $\log1/\epsilon).$
\end{proof}

\begin{cor}
If $\FS$ is an {\rm MQ}, the complexity is $
\widetilde O(m^{5.5}(n\log p+r)^{2.5}(n+T_\FS)\log^{2} p\kappa^2\log1/\epsilon).$
\end{cor}

\begin{cor}
If $q=2^m$, then the complexity is $\widetilde O(m^{5.5}T_\FS^{3.5}D^{3.5}\kappa^2\log1/\epsilon)$.
Moreover, if $\FS\subset\F_{2^m}[\X]$ is an {\rm MQ},  then the complexity is $\widetilde O((n+r)^{2.5}(n+T_\FS)\kappa^2\log1/\epsilon)$.
\end{cor}

\section{Reduce inequalities to MQ in Boolean variables}
\label{sec-MQB}

In this section, we show how to reduce the inequality constraints
 $\IS=\{0\le g_i(\X,\Y)\le b_i, i$ $=1,\ldots,s; 0\le y_k\le u_k,k=1,\ldots,m; \X\in\F_p^n; \Y\in\Z^m \}$ of problem \bref{eq-op} into a B-POSSO, where $g_1,\ldots,g_s\in\Z[\X]$,
$b_1,\ldots,b_s,u_1,\ldots,u_m\in\N$.
%
%
We emphasize that for  $g\in\C[\X,\Y]$, $\check{\X}\in\F_p^n$,
and $\check{\Y}\in\Z^m$, $g(\check{\X},\check{\Y})$ is evaluated in $\C$.

\subsection{Reduce polynomial system over $\C$ to MQ in Boolean variables over $\C$}
\label{sec-ie2}
Let $\GS = \{g_1,\ldots,g_s\}\subset\Z[\X,\Y]$. We will reduce
$\GS$ into an equivalent MQ in Boolean variables over $\C$
under the condition $x_i\in\F_p=\{0,1,\ldots,p-1\}$ and
$0\le y_j \le u_j$.
Let $d_g = \max_{l=1}^s \deg(g_l)$.

Following Lemma \ref{lm-Q}, let $Q(\GS)\subset\Z[\X,\Y,\V]$ be the MQ
defined in \bref{eq-Q}, where $\V$ is the set of new indeterminates introduced in Lemma \ref{lm-Q}.
We will reduce $\X$, $\Y$, and $\V=\{v_{1},\ldots,v_{l}\}$ to Boolean variables.
For $\X$, we use \bref{eq-xbit} to rewrite them as Boolean variables $\X_\bit$.
For $\Y$, using Lemma \ref{lm-bb}, the integers $y_i$ satisfying $0\le y_i\le u_i$
can be represented exactly as follows
\begin{eqnarray}
y_i &=&  \theta_{u_i}(\mathbb{Y}_i) =
  \sum_{j=0}^{\lfloor\log_2u_i\rfloor-1}Y_{i,j}2^j+
(u_i-2^{\lfloor\log_2u_i)\rfloor}+1)Y_{i,\lfloor\log_2u_i\rfloor},\cr
\mathbb{Y}_i&=&\{Y_{i,j}\,|\,j=0,\ldots, \lfloor\log_2u_i\rfloor\},\label{eq-ybit}\\
\Y_\bit&=&\cup_{i=1}^m\Y_i\nonumber
\end{eqnarray}
where $Y_{i,j}$ are Boolean variables.

From Lemma \ref{lm-Q}, each $v_{i}\in\V_k$ represents a monomial in $\X$
and $\Y$ of degree $\le d_g$. So,
$0\le v_{i}\le h^{d_g}$ for $h=\max\{{p-1}, u_1,\ldots,u_m\}$.
By Lemma \ref{lm-bb}, we can write $v_{i}$ as
\begin{equation}\label{eq-vbit}
v_{i}=  \theta_{h^{d_g}}(\mathbb{V}_{i,\bit}) =
  \sum_{j=0}^{\lfloor d_g\log_2h\rfloor-1}V_{i,j}2^j+
(h^{d_g}-2^{\lfloor d_g \log_2h\rfloor}+1)V_{i,\lfloor d_g\log_2h\rfloor},
\end{equation}
where $\mathbb{V}_{i,\bit}=\{V_{i,j}\,|\,j=0,\ldots,
 \lfloor d_g \log_2u_i\rfloor\}$,
${\V}_{\bit}=\cup_{i=1}^t\V_{i,\bit}$, and each $V_{i,j}$ is a Boolean variable.

Let $\hat{g}_k=\widehat{Q}(g_k)\in \Z[\X,\Y,\V]$  be defined  in  \bref{eq-Qf},
$\overline{Q}(\GS)= Q(\GS)\setminus\{\hat{g}_1,\ldots,\hat{g}_s\}$.
Substituting $x_i$ in \bref{eq-xbit},
and $y_i$ in \bref{eq-ybit}, and
$v_{i}$ in \bref{eq-vbit}
into ${Q}(\GS)$, $\hat{g}_k$, and $\widetilde{Q}(\GS)$, we obtain
\begin{eqnarray}\label{eq-ybit1}
B(\GS), \overline{g}_k, \overline{B}(\GS)  \hbox{ in } \Z[\X_\bit,\Y_\bit,{\V}_{\bit}].
\end{eqnarray}
The following result shows that $\GS$ and $B(\GS)$ are equivalent.
\begin{lem}\label{lm-gs1}
For $\check{\X}\in\F_p^n$ and $\check{\Y}\in\Z^n$ such that $0\le \check{y}_j \le u_j$ for each $j$, there exists a $\check{\V}_{\bit}$ such that
${g}_k(\check{\X},\check{\Y}) = \overline{g}_k(\check{\X}_\bit,\check{\Y}_\bit,\check{\V}_{\bit})$
for $k=1,\ldots,s$ and $\overline{B}(\GS)(\check{\X}_\bit,\check{\Y}_\bit,\check{\V}_{\bit})=0$.
\end{lem}
\begin{proof}
From \bref{eq-Q}, it is easy to see that starting from $\check{\X}\in\F_p^n$
and $\check\Y\in\Z^m$, one may obtain a unique $\check{\V}$ such that $\widetilde{B}(g_k)(\check{\X},\check{\Y},\check{\V})=0$
and $g_k(\check{\X},\check{\Y}) = \hat{g}_k(\check{\X},\check{\Y},\check{\V})$ for each $k$.
It suffices to show that $\check{\X},\check{\Y},\check{\V}$
can be written as their Boolean forms, which is
valid for $\check{\X},\check{\Y}$ due to \bref{eq-xbit} and \bref{eq-ybit}
and Lemma \ref{lm-bb}.
From Lemma \ref{lm-Q}, each $v_{i}\in\V$ is a monomial in $\X$
and $\Y$ of degree $\le d_g$. So,
$0\le v_{i}\le h^{d_g}$ for $h=\max\{{p-1}, u_1,\ldots,u_m\}$.
By Lemma \ref{lm-bb}, there exists a $\check{V}_{\bit}$
such that $g_k(\check{\X},\check{\Y}) = \hat{g}_k(\check{\X},\check{\Y},\check{\V})
= \overline{g}_k(\check{\X}_\bit,\check{\Y}_\bit,\check{\V}_{\bit})$
and $\overline{B}(g)(\check{\X}_\bit,\check{\Y}_\bit,\check{\V}_{\bit})=0$.
\end{proof}

\begin{lem}\label{lm-gs2}
$B(\GS)= \{\overline{g}_k\}\cup \overline{B}(\GS)$
defined in \bref{eq-ybit1}
has    $O((m+n) T_\GS d_g\log d_g \log h)$ number of variables
and total sparseness  $O((m+n)T_\GS d_g^2\log d_g \log^2 h)$,
and  $C(B(\GS))$ is $O(C(\GS)+ d_g\log h)$,
where $h=\max\{{p-1}, u_1,\ldots,u_m\}$ and $d_g = \max_{l=1}^s \deg(g_l)$..
\end{lem}
\begin{proof}
By Lemma~\ref{lm-Q},
$N_{Q(\GS)}=O((m+n)T_\GS\log d_g)$,
$T_{Q(\GS)}=O((m+n)T_\GS\log d_g)$,
and $C(Q(\GS)) = C(\GS)$.
Note that $|\X| =n, |\Y| = m$, and
$|\V|$ is bounded by $N_{Q(\GS)}=O((m+n)T_\GS\log d_g)$.
By \bref{eq-xbit}, $|\X_{\bit}|=O(n\log p)=O(n\log h)$.
By \bref{eq-ybit}, $|\Y_{\bit}|=O(m\log h)$.
By \bref{eq-vbit} and Lemma \ref{lm-bb},
$|\overline{\V}_\bit| =O((m+n)T_\GS  \log d_g \log h^{d_g})) =O((m+n)d_g T_\GS\log d_g\log h))$.
By \bref{eq-yt2} and Lemma \ref{lm-bb}, we have $\#\G_\bit={O}(s\log b)$.
Since $s\le T_\GS$, $p-1,b\le h$, the number of Boolean variables are
$N_{B(\GS)}=O((m+n)d_g T_\GS \log d_g \log h)$.

Note that monomials of $Q(\GS)$ are of the form $x_ix_j, x_iy_j, x_iv_j, y_iy_j, y_iv_j$, or $v_iv_j$ when we rewrite them as Boolean variables, the sparseness of the new expressions are bounded by
$\log^2 p$, $\log p\log h$, $d_g\log p \log h $,  $\log^2 h$, $d_g\log^2 h$ and $d_g^2\log^2 h$, respectively.
The total sparseness of ${B}(\GS)$
is $O((m+n)T_\GS d_g^2 \log d_g \log^2 h)$.

From \bref{eq-vbit} and the fact that $B(\GS)$ is MQ,
the bit size of the coefficients of ${B}(\GS)$ is $O(C(\GS)+d_g\log h )$.
\end{proof}

\begin{rem}\label{rem-xbit}
For inequalities involving variables over finite fields, the solution
of the inequalities depends on the representation of $\F_p$.
For the general optimization problem \ref{eq-op}, we just use standard representation for $\F_p$.
For specific problems, such as the SIS problem in Section \ref{sec-sis},
we use different representations for $\F_p$ to find the ``correct" solution.
\end{rem}

\subsection{Reduce inequalities into MQ in Boolean variables}
\label{sec-ie3}
We now consider the inequality constraints of problem \bref{eq-op}:
$\IS=\{0\le g_i(\X,\Y)\le b_i, i=1,\ldots,s; 0\le y_k\le u_k,k=1,\ldots,m; \X\in\F_p^n \}$, where $g_1,\ldots,g_s\in\Z[\X,\Y]$,
$b_1,\ldots,b_s,u_1,\ldots,u_m\in\N$.
We will reduce $\IS$ into an MQ in Boolean variables.
%
%
Let
\begin{eqnarray}\label{eq-yt2}
\mathbb G_i&=&\{G_{i,k}\,|\,k=0,\ldots,\lfloor\log_2b_i\rfloor\}, \G_\bit=\cup_{i=1}^s\G_i,\cr
\delta(g_i) &=&\theta_{b_i}(\mathbb G_i) -\overline{g}_i= \sum_{k=0}^{\lfloor\log_2b_i\rfloor-1}G_{i,k} 2^k+(b_i-2^{\lfloor\log_2b_i\rfloor}+1)G_{i,\lfloor\log_2b_i\rfloor} -\overline{g}_i  
\\
I(\IS) &=& \{\delta(g_1),\ldots,\delta(g_s)\} \cup \overline{B}(\GS)
\subset\Z[\X_\bit,\Y_\bit,{\V}_\bit,\G_\bit]\nonumber
\end{eqnarray}
where  $G_{i,k}$ are Boolean variables,
$\overline{g}_i$ and $\overline{B}(\GS)$  are defined in \bref{eq-ybit1}.
We summarize the result of this section as the following result.
\begin{lem}\label{lm-is1}
$\check{\X}\in\F_p^n$ and $\check{\Y}\in\Z^m$ satisfy the constraint $\IS$ if and only if
there exist Boolean values $\check{\V}_\bit,\check{\G}_\bit$ such that
$(\check{\X}_\bit,\check{\Y}_\bit,\check{\V}_\bit,\check{\G}_\bit)$  is a solution of $I(\IS)$.
\end{lem}
\begin{proof}
By Lemma \ref{lm-bb},  $0\le \overline{g}_i(\X_\bit,\Y_\bit,\V_\bit)\le b_i$ if and only if
$\exists \G_\bit$ such that  $\delta(g_i)(\X_\bit,\Y_\bit,\V_\bit,$ $\G_\bit)=0$.
Then, the lemma is a consequence of Lemma~\ref{lm-gs1}.
%
%
%
\end{proof}

We now estimate the parameters of $I(\IS)$.
Let $b=\max_{i=1}^s b_i$, $d_g=\max_{i=1}^s \deg(g_i)$, $h=\max\{p-1, b, u_1,\ldots,u_m\}$,
$\GS=  \{ g_1,\ldots,g_s \}$
and $T_\GS \ge s$ the total sparseness of $\GS$.
Then, we have
\begin{lem}\label{lem-ineq2}
$I(\IS)$ has $O((m+n) T_\GS d_g\log d_g \log h)$  variables
and total sparseness $O((m+n)T_\GS d_g^2$ $\log d_g \log^2 h)$.
$C(I(\IS))$ is $O(C(\GS)+ d_g\log h )$ .
\end{lem}
\begin{proof}

Since $B(\GS)= \{\overline{g}_k\}\cup \overline{B}(\GS)$,
from \bref{eq-yt2},
$N_{I(\GS)} = N_{B(\GS)} + \#\G_\bit$,
$T_{I(\GS)} = T_{B(\GS)} + \sum_i\#\theta_{b_i}(\G_i)= T_{B(\GS)} + \#\G_\bit$,
and
$C(I(\GS)) = C(B(\GS))$.
Note that $\#\G_\bit = O(s\log b)$.
Since $T_\GS \ge s$,  $\#\G_\bit$ is negligible comparing to
the complexity of  ${B(\GS)}$
and the lemma follows directly from Lemma \ref{lm-gs2}.
\end{proof}

From Lemma \ref{lem-ineq2}, the total sparseness and the coefficients  of $I(\GS)$
are well controlled.

\begin{cor}\label{cor-lini}
If $g_i$ are linear, then
$I(\GS)\subset\Z[\X_\bit,\Y_\bit,\G_\bit]$,
$T_{ {B}(\GS)}=O(T_\GS\log h)$, and
$N_{B(\GS)}=(n+m)\log h$.
Furthermore,
$T_{I(\IS)}=O(T_\GS\log h+s\log b)$ and
$N_{I(\IS)}= O((n+m)\log h+s\log b)$.
\end{cor}
\begin{proof}
Since each $g_i$ is linear, we have $Q(f_i)=f_i$.
Then the variable $V_{k,i,j}$ are not needed and
$\overline{B}(\GS)$ has $(n+m)\log h$ indeterminates.
Also, $T_{\overline{B}(\GS)}=O(T_\GS\log h)$.
The results for $I(\GS)$ can be proved similarly.
\end{proof}

\subsection{Bounded integer solutions of polynomial inequalities and equations}
As a direct application of the reduction method given in this section,
we can give a quantum algorithm to find a feasible solution
to the inequality constraint
$\IS=\{0\le g_i(\X,\Y)\le b_i, i=1,\ldots,s; 0\le y_k\le u_k,k=1,\ldots,m; \X\in\F_p^n \}$, where $g_1,\ldots,g_s\in\Z[\X,\Y]$,
$b_1,\ldots,b_s,u_1,\ldots,u_m\in\N$.
Use the notations in Lemma \ref{lem-ineq2}, we have
\begin{prop}\label{cor-fsol1}
For $\epsilon\in(0,1)$, there is a quantum algorithm to compute a feasible solution to $\IS$ with probability $>1-\epsilon$ and
in time $\widetilde O((m+n)^{3.5}T_\GS^{3.5} d_g^{4.5}\log^{4.5} h \kappa^2\log1/\epsilon)$, where $\kappa$ is the condition number of $I(\IS)$
defined in \bref{eq-yt2}.
\end{prop}
\begin{proof}
By Lemma \ref{lm-is1}, to find a feasible solution to $\IS$,
we need only to find a Boolean solution of $I(\IS)$.
By Lemma \ref{lem-ineq2},
$N_{I(\IS)} =O((m+n) T_\GS d_g\log d_g \log h)$,
$T_{I(\IS)} =O((m+n)T_\GS d_g^2\log d_g \log^2 h)$.
Since $N_{I(\IS)}< T_{I(\IS)}$,
by Theorem \ref{th-m2}, the complexity to find a Boolean solution of
$I(\IS)$ is
$\widetilde O(N_{I(\IS)}^{2.5}T_{I(\IS)}\kappa^2$ $\log1/\epsilon)
=\widetilde O((m+n)^{3.5}T_\IS^{3.5} d_g^{4.5}\log^{4.5} h \kappa^2\log1/\epsilon)$.
\end{proof}

A closely related problem is to find bounded integer solutions
of a polynomial system over $\Z$.
\begin{prop}\label{cor-fsol2}
Let $\GS=\{g_1,\ldots,g_s\}\subset \Z[\Y]$ and $\epsilon\in(0,1)$.
There is a quantum algorithm to compute an integer solution $\mathbf b=(b_1,\ldots,b_m)$ of $\GS=0$ satisfying  $0\le b_i\le u_i$ for each $i$
with probability $>1-\epsilon$ and
in time $\widetilde O(m^{3.5}T_\GS^{3.5} d_g^{4.5}\log^{4.5} h \kappa^2\log1/\epsilon)$, where $\kappa$ is the condition number of
$\widetilde{B}(\GS)$ to be defined in the proof and
$h = \max_i{u_i}$.
\end{prop}
\begin{proof}
By Lemma \ref{lm-gs1}, to find an integer solution to $\GS=0$,
we need just to find a Boolean solution of
${B}(\GS)$ defined in   \bref{eq-ybit1}.
By  Lemma \ref{lm-gs2},
we have
$N_{ {B}(\GS)}= O(m T_\GS d_g\log d_g \log h)$
and
$T_{ {B}(\GS)}=O(m T_\GS d_g^2\log d_g \log^2 h)$.
Since $N_{ {B}(\GS)}< T_{ {B}(\GS)}$,
by Theorem \ref{th-m2}, the complexity to find a Boolean solution of
$ {B}(\GS)$ is
$\widetilde O(N_{ {B}(\GS)}^{2.5}T_{ {B}(\GS)}\kappa^2\log1/\epsilon)
=\widetilde O(m^{3.5}T_\GS^{3.5} d_g^{4.5}\log^{4.5} h \kappa^2\log1/\epsilon)$.
\end{proof}

For a general polynomial system in $\C[\X]$, the bound for coordinates of solutions could be double-exponential, as shown by the following example.
\begin{exmp}\label{ex-de}
For $\FS=\{x_1-2,x_2-x_1^2,x_3-x_2^2,\ldots,x_n-x_{n-1}^2\}\subset\C[\X]$, $\V_\C(\FS)=\{(2,2^2,2^4,\ldots,$ $2^{2^{n-1}})\}$.
\end{exmp}
On the other hand,
the isolated solutions of a polynomial system is at most double-exponential
\cite[p.~341]{yap1}. In a similar way, it is also possible
to find bounded rational solutions of a polynomial system.

\section{Optimization over finite fields}
\label{sec-op}

\subsection{A quantum algorithm for the optimization problem}

\label{sec-op2}
In this section, we give a quantum algorithm to solve
the  optimization problem \bref{eq-op}. The idea is to search
the minimal value of the objective function by
solving several B-POSSOs, which will be done in four steps.

{\bf Step 1}. By Lemmas \ref{lm-Q} and \ref{lm-crct}, we  reduce
the equational constraints $f_j(\X)=0 \mod p, j=1,\ldots,r$
to an MQ in Boolean variables over $\C$:
$\FS_1=P(Q(\FS))\subset\C[\X_\bit,\V_{1\bit},\U_\bit]$.

{\bf Step 2}. By Lemma \ref{lm-is1},
we reduce the inequality constraints $\IS=\{0\le g_i(\X,\Y)\le b_i, i=1,\ldots,s\}$
to an MQ in Boolean variables over $\C$:
$\GS_1=I(\IS)\subset\C[\X_\bit,\Y_\bit,\V_{2\bit},\G_{\bit}]$.

{\bf Step 3}.
Applying Lemma \ref{lm-gs1} to the objective function $o(\X,\Y)$,
we may reduce $o$ into a quadratic polynomial in Boolean variables
$\overline{o}\in\C[\X_\bit,\Y_\bit,{\V}_{3\bit}]$
and an MQ
$\GS_2=\overline{B}(\{o\})\subset\Z[\X_\bit,\Y_\bit,
{\V}_{3\bit}]$
defined in \bref{eq-ybit1}.
For the simplicity of presentation,
we denote $\V_\bit =  \V_{1\bit}\cup {\V}_{2\bit}\cup {\V}_{3\bit}$.
Let
\begin{eqnarray}\label{eq-oCS}
\CS &=& \FS_1\cup\GS_1\cup\GS_2\subset
\C[\X_\bit,\Y_\bit,\V_\bit,\U_\bit,\G_\bit].
\end{eqnarray}

A $(0,1)$-programming is an optimization problem
where all the arguments take values of 0 or 1.
By  Lemmas \ref{lm-Q}, \ref{lm-crct}, \ref{lm-gs1}, and \ref{lm-is1}, we have
\begin{lem}
Problem \bref{eq-op} is equivalent to the following nonlinear $(0,1)$-programming problem
\begin{eqnarray}\label{eq-op1}
\min_{\W_\bit}  \overline{o}(\W_\bit)&&\hbox{ subject to }
 \CS(\W_\bit)=0
\end{eqnarray}
where $\W_\bit =(\X_\bit,\Y_\bit,\V_\bit,\U_\bit,\G_\bit)$ and $\CS$ is defined in \bref{eq-oCS}.
\end{lem}


{\bf Step 4}.
The basic idea to search a minimal value of the objective function is as follows.
Since all the variables are bounded, the objective function is also
 bounded, so we may assume  $\alpha\le \overline{o}(\check\Z_\bit) < \mu$ for some
$\alpha,\mu\in \N$.
We divide $[\alpha,\mu)$  into two roughly equal parts: $[\alpha,\alpha+2^\beta)$ and $[\alpha+2^\beta,\mu)$
and solve the following decision problem
\begin{equation}\label{eq-ab1}
\exists \W_\bit(\overline{o}(\W_\bit)\in[\alpha,\alpha+2^\beta)
 \hbox{ and }  (\CS(\W_\bit)=0)).
\end{equation}
Let
\begin{eqnarray}
\delta_{\alpha\beta}(\overline{o})
&=&
\alpha+\sum_{j=0}^{\beta-1}F_j2^j-\overline{o}(\W_\bit)
 \in\Z[\Z_\bit],\label{eq-od}\\
L_{\alpha\beta} &=& \CS\cup\{\delta_{\alpha\beta}(\overline{o})\}
\subset\Z[\Z_\bit],\label{eq-L}
\end{eqnarray}
where $\Z_\bit=\W_\bit\cup \F_\bit=\{\X_\bit,\Y_\bit,\V_\bit,\U_\bit,\G_\bit,\F_\bit\}$ and
$\F_\bit=\{F_0,\ldots,F_{\beta-1}\}$ are Boolean variables.
By Lemma \ref{lm-bb}, we have
\begin{lem}\label{lm-ln01}
Problem \bref{eq-ab1} has a solution $\check{\W}_\bit$ if and only if
$L_{\alpha\beta}=0$ has a solution $\check{\Z}_\bit=(\check{\W}_\bit,\check{\F}_\bit)$.
\end{lem}
%
If the answer to problem \bref{eq-ab1} is yes, we repeat the procedure for
the new feasible interval $[\alpha,\overline{o}(\check{\Z}_\bit))$.
If the answer is no, we repeat the procedure for
the new feasible interval $[\alpha+2^\beta,\mu)$.
The procedure ends when $\mu=\alpha+1$.

We now give the algorithm to solve problem \bref{eq-op}.
For convenience of later usage, we add a new constraint $0\le o< u$
for a given $u\in\N{>0}$.
\begin{alg}[QFpOpt]\label{alg-opt}
\end{alg}
{\noindent\bf Input:}
Problem \bref{eq-op}, $\epsilon\in(0,1)$, and a $u\in\Z_{>0}$ such that $0\le o< u$.

{\noindent\bf Output:} $\widehat{o}, \check{\X}\in\F_p^n,$ and $\check{\Y}\in\Z^m$
such that $\widehat{o}= o(\check{\X},\check{\Y})$ is the minimal
value of $o$, or ``fail''.

\begin{description}

\item[Step 1:]
Set $\alpha=0,\mu=u$.

\item[Step 2:]
Compute $\CS$ in \bref{eq-oCS}.

\item[Step 3:]
Let $\beta=\lfloor\log_2 (\mu-\alpha)\rfloor-1$
and compute
  $L_{\alpha\beta}\subset\C[\Z_\bit]$ defined in \bref{eq-L}.

\item[Step 4:]
Let $\check{\Z}_\bit=$
{\bf QBoolSol}$(L_{\alpha\beta},\epsilon/\log_{4/3}u)$,
where {\bf QBoolSol} is from Theorem \ref{th-m2}.

\item[Step 5:]
If Algorithm {\bf QBoolSol} returns a solution:
$\check{\Z}_\bit=\{\check{\X}_\bit,\check{\Y}_\bit,\check\V_\bit,\check\U_\bit,
\check{\G}_\bit,\check{\F}_\bit\}$, then
\begin{description}
\item[Step 5.1:]
Compute $\check\X$ and $\check{\Y}$ from $\check\X_\bit$
and $\check{\Y}_\bit$ according to \bref{eq-xbit} and \bref{eq-ybit}, respectively.
\item[Step 5.2:]
If $\check{\F}_\bit=\mathbf 0$, return $\alpha$, $\check\X$ and $\check{\Y}$.
\item[Step 5.3:]
If $\check{\F}_\bit\ne\mathbf 0$, let $\mu=\overline{o}(\check{\Z}_\bit)$ and goto Step 3.
\end{description}

\item[Step 6:]
If  $\check{\Z}_\bit=\emptyset$, then
\begin{description}
\item[Step 6.1:]
If $\mu-\alpha>1$, let $\alpha=\alpha+2^\beta$, and goto Step 3.

\item[Step 6.2:]
If $\mu-\alpha=1$ and $\mu\ne u$, return $\mu$, $\check\X$ and $\check{\Y}$.
\item[Step 6.3:]
If $\mu-\alpha=1$ and $\mu= u$, return ``fail''.

\end{description}
\end{description}

Let $b=\max_{i=1}^s b_i$, $d_f=\max_{i,j}\{2, \deg(f_i,x_j)\}$, $d_g=\max_{i,j}\{2, \deg(g_i,x_j)\}$, $h=\max\{p-1, u_1,\ldots,$  $u_m\}$, and $\GS_o=  \{o, g_1,\ldots,g_s \}$.
Then, we have
\begin{thm}\label{th-opt1}
{Algorithm \ref{alg-opt}} gives a solution to problem \bref{eq-op}
with constraint $0\le o< u$
with success probability $\ge 1-\epsilon$ and in  complexity
$\widetilde O(N_{L_{\alpha\beta}}^{2.5}T_{L_{\alpha\beta}}\kappa^2\log(1/\epsilon) \log u)$,
where
\begin{eqnarray*}
N_{L_{\alpha\beta}} &=& \widetilde{O}(n T_\FS \log d_f\log p
+ (m+n)T_{\GS_o} d_g \log h+\log u),\\
T_{L_{\alpha\beta}}&=&\widetilde{O}(nT_\FS \log d_f \log^2 p +
(m+n)T_{\GS_o}d_g^2\log^2 h+\log u),
\end{eqnarray*}
and $\kappa$ is the maximal condition number of all $L_{\alpha\beta}$ in the algorithm, called the {\em  condition number} of the problem.
\end{thm}
\begin{proof}
We first prove the termination of the algorithm by showing
that the feasible interval $[\alpha,\mu)$ will decrease strictly after each loop starting from Step 3.
In Step 3, we split $[\alpha,\mu)=[\alpha,\alpha+2^\beta)\cup[\alpha+2^\beta,\mu)$ with $(\mu-\alpha)/4<2^\beta\le(\mu-\alpha)/2$.
In Step 5.3,  we start a new loop for $[\alpha,\mu_1)$, where $\mu_1=\overline{o}(\check{\Z}_\bit)< 2^\beta$. Then after this step,
the feasible interval will decrease by at least $\frac12(\mu-\alpha)$ due to $2^\beta\le(\mu-\alpha)/2$.
In Step 6.1, we start a new loop for $[\alpha+2^\beta,\mu)$.
After this step, the feasible interval will decrease by more than $\frac14(\mu-\alpha)$ due to $(\mu-\alpha)/4<2^\beta$.
In summary, after each loop, the algorithm either terminates
or has a smaller feasible interval which is of at most $3/4$ of the size of the feasible
interval of the previous loop.
So, the algorithm will terminate after at most $\log_{4/3}u$ loops.

We now prove the correctness of the algorithm, which follows from the following claim:
\begin{equation}\label{eq-claim}
\hbox{The minimal value of } o \hbox{ is in } [\alpha,\mu] \hbox{ during the algorithm}
\end{equation}
if the minimal value exists and Algorithm {\bf QBoolSol} in
Step 4 always returns a solution of $L_{\alpha\beta}$ if such a solution exists.
The above claim is obviously true for the initial values given in Step 1.

In Step 5, by Lemma \ref{lm-ln01}, we find a solution $\check{\Z}_\bit$ such that $\overline{o}(\check{\Z}_\bit) \in [\alpha,\alpha+2^\beta)$.
In Step 5.2, the condition $\check{\F}_\bit=\mathbf 0$ means that $\overline{o}(\check{\Z}_\bit) =\alpha$ and
the minimal solution $\overline{o}$ is found by   claim \bref{eq-claim}.
In Step 5.3, the condition $\check{\F}_\bit\ne\mathbf 0$ means that $\overline{o}(\check{\Z}_\bit)\ne\alpha$ and
we have a new $\mu_1=\overline{o}(\check{\Z}_\bit)$.
%
Since $o\in[\alpha,\mu_1]$ has a solution $\check{\Z}_\bit$,
by   claim \bref{eq-claim}, the minimal value of $o$ is in  $[\alpha,\mu_1]$,
and the claim is proved in this case.

In Step 6, {\bf QBoolSol} returns $\emptyset$,  meaning that
$\overline{o}(\check{\Z}_\bit) \in [\alpha,\alpha+2^\beta)$ has no solution and the minimal value of $o$ must be in $[\alpha+2^\beta,\mu)$ if it exists.
So, in Step 6.1, we will  find the minimal value of $o$ in $[\alpha+2^\beta,\mu)$ in the next loop, and claim \bref{eq-claim} is proved in this case.
In Step 6.2, we have  $\mu-\alpha=1$ and $\mu\ne u$.
Since $o\in [\alpha,\alpha+2^\beta)$ has no solution,
by claim \bref{eq-claim}, $\mu=\alpha+1$ must be the
minimal value of $o$.
%
Note that in Step 6, we only update the lower bound $\alpha$.
In Step 5, we only  update the upper bound $\mu$,
and when $\mu$ is updated we have $\mu=\overline{o}(\check{\Z}_\bit) ={o}(\check{\X},\check{\Y})$.
Therefore, $\mu={o}(\check{\X},\check{\Y})$ is always valid, once Step  5 is executed.
The condition $\mu\ne b$ implies that Step 5 has been executed at least one time
and hence $\mu={o}(\check{\X},\check{\Y})$.
In Step 6.3, the conditions $\mu-\alpha=1$ and $\mu= u$
means that Step 5 is never executed and the problem has no solution.

Finally, the solution obtained by {Algorithm \ref{alg-opt}} is correct if and only if each Step 4 is correct, that is, if $L_{\alpha\beta}$ does have solutions, then {\bf QBoolSol} will return a solution.
Since Step 4 will execute at most $\log_{4/3}u$ times,
by Theorem \ref{th-m2},
the probability for the algorithm to be correct is at least $(1-\epsilon/\log_{4/3}u)^{\log_{4/3}u}>1-\epsilon$.

We now analyse the complexity.
Note that $2$ is added to $d_f$ to make sure $\log d_f\ne0$.
By Lemma \ref{lm-fsb},
$\FS_1$ is of total sparseness $O(nT_\FS \log d_f \log^2 p)$ and has  $O(n T_\FS \log d_f \log p)$ indeterminates.

By Lemma \ref{lem-ineq2},
$\GS_1=I(\IS)$
is of total sparseness  $O((m+n)T_\GS d_g^2\log d_g\log^2 h)$
and has $O((m+n)d_g T_\GS \log d_g \log h)$
indeterminates.
Also,
$\GS_2=\overline{B}(o)$ is of
total sparseness  $O((m+n)T_o d_g^2 \log d_g \log^2 h)$
and has $O((m+n)d_g T_o \log d_g \log h)$ indeterminates.

$\delta_{\alpha\beta}(\overline{o})$ is of total sparseness
$O(\log u +T_{\overline{o}})
= O(\log u + T_o d_g^2\log^2 h)$
and has
$O(\log u + (m+n)d_g T_o \log d_g \log h)$
indeterminates, since $2^\beta-\alpha < u$.

Then,
$L_{\alpha\beta} = \CS\cup\{\delta_{\alpha\beta}(\overline{o})\}= \FS_1\cup\GS_1\cup\GS_2\cup\{\delta_{\alpha\beta}(\overline{o})\}$ is of total sparseness
$T_{L_{\alpha\beta}}=T_{\FS_1} + T_{\GS_1} +T_{\GS_2} + T_{\delta_{\alpha\beta}(\overline{o})}=
\widetilde{O}(nT_\FS \log d_f\log^2 p +
(m+n)T_{\GS_o} d_g^2\log^2 h+\log u)$
and has
$N_{L_{\alpha\beta}}=\widetilde{O}(n T_\FS \log d_f\log p
+(m+n)d_g T_{\GS_o}\log h+\log u)$
indeterminates.

In Step 4, by {Theorem \ref{th-m2}}, since $N_{L_{\alpha\beta}} < T_{L_{\alpha\beta}}$, we can find a Boolean solution of $L_{\alpha\beta}$ in time
$\widetilde{O}(N_{L_{\alpha\beta}}^{2.5}(N_{L_{\alpha\beta}}+
T_{L_{\alpha\beta}})\kappa^2\log(\epsilon/\log u))
= \widetilde{O}(N_{L_{\alpha\beta}}^{2.5}T_{L_{\alpha\beta}}
\kappa^2\log(\epsilon/\log u)).$
It is clear that in each loop, the complexity of the algorithm is dominated by Step 4.
Since we have at most $\log_{4/3}u$ loops,
 the complexity for {Algorithm \ref{alg-opt}} is
$\widetilde O(N_{L_{\alpha\beta}}^{2.5}T_{L_{\alpha\beta}}\kappa^2 \log((\log_{4/3}u)/\epsilon) \log_{4/3}u)=
\widetilde O(N_{L_{\alpha\beta}}^{2.5}T_{L_{\alpha\beta}}
\kappa^2\log(1/\epsilon)$ $ \log u)$.
\end{proof}

We now show how to solve the original problem \bref{eq-op}.
\begin{cor}
{Algorithm \ref{alg-opt}} gives a solution to problem \bref{eq-op}
with the same probability and complexity for $u=2\#(o) h_o h^{d_o}+1$,
where $h_o$ is the height of the coefficients of $o$, $h=\max\{p-1, u_1,\ldots,u_m\}$, and $d_o=\deg(o)$.
\end{cor}
\begin{proof}
It is easy to see that $|o| \le \#(o) h_o h^{d_o}$,
so $0\le \widetilde o  < u$ for the new objective function
$\widetilde o =o+ \#(o) h_o h^{d_o}$. Then Theorem \ref{th-opt1}
can be used to the new optimization problem.
\end{proof}

\begin{rem}
Note that the upper bound $u=2\#(o) h_o h^{d_o}+1$ for the objective function is quite large. An alternative way is to use {\bf Algorithm QBoolSol} in Theorem \ref{th-m2} to find a solution
$\check{\X}_\bit,\check{\Y}_\bit$
for $\FS_1\cup\GS_1\subset\C[\X_\bit,\Y_\bit,\overline\V_{1\bit},
\overline\V_{2\bit},\U_\bit,\G_\bit]$
and set $u=2o(\check{\X},\check{\Y})+1$.
Then for the new objective function $\widetilde o =o  + o(\check{\X},\check{\Y})$,
we can use the constraint $0\le \widetilde o  < u$ to find a solution to problem \bref{eq-op}.
This does not change the complexity of the algorithm.
\end{rem}

\subsection{Applications to linear $(0,1)$-programming and QUBO}
In this section, we use two $(0,1)$-programming problems  to illustrate Algorithm \ref{alg-opt}.
QUBO means {\em quadratic unconstrained binary optimization},
which is the mathematical problem that can be solved by the  D-Wave

The linear $(0,1)$-programming is one of Karp's 21 NP-complete problems~\cite{Karp}
which covers lots of fundamental computational problems, such as
the subset sum problem,
the assignment problem,
the traveling salesperson problem,
the knapsack problem, etc.
For more information about this problem, please refer to~\cite{lip-survey}.
The linear $(0,1)$-programming can be stated as follows~\cite{Balas}
\begin{eqnarray}\label{eq-01}
&&\min_{\Y_\bit\in\{0,1\}^m} o(\Y_\bit) = \sum_{j=1}^m c_j y_j\hbox{ subject to }
 \quad \sum_{j=1}^m a_{ij} y_j \le h_i , i=1,\ldots, s
\end{eqnarray}
where  $\Y_\bit=(y_1,\ldots,y_m)$ and $a_{ij}, c_j,h_i\in \Z$ for any $i,j$.
We reduce problem \bref{eq-01} to the standard form \bref{eq-op}.
Let $e_i = \sum_{j=1}^m |a_{ij}| \in \Z_{\ge0},\  1\le i\le s$
and $g_i = \sum_{j=1}^m a_{ij}y_j + e_i$
and  $b_i=h_i+e_i, 1\le i\le s$.
Let $u= 2\sum_{i=1}^m |c_j|+1 \in\N$.
Then, problem \bref{eq-01} is equivalent to
\begin{eqnarray}\label{eq-012}
&&\min_{\Y_\bit\in\{0,1\}_2^m} o_B(\Y_\bit) = \sum_{j=1}^m c_j y_j + (u-1)/2\hbox{ subject to }\\
&&\quad 0\le o_B(\Y_\bit)< u; 0\le g_i \le b_i, i=1,\ldots, s.\nonumber
\end{eqnarray}

%
So we can use Algorithm \ref{alg-opt} to solve problem \bref{eq-012}.
%
%
Let $\GS=\{g_1,\ldots,g_s\}$.
Since $g_i$ are linear, we do not need to compute $Q(g_i)$
and  $\overline{\V}_\bit=\emptyset$ (see \bref{eq-vbit} for definition)
and $\overline{B}(\GS)=\emptyset$ (see \bref{eq-ybit1} for definition).
Since $y_j$ are Boolean variables, we do not need to use
\bref{eq-ybit} to expand them and hence $\overline{g}_i = g_i$.
So, \bref{eq-yt2} becomes
\begin{eqnarray*}\label{eq-yt01}
\delta(g_i) &=&\theta_{b_i}(\mathbb G_i)= \sum_{k=0}^{\lfloor\log_2b_i\rfloor-1}G_{i,k} 2^k+(b_i-2^{\lfloor\log_2b_i\rfloor}+1)G_{i,\lfloor\log_2b_i\rfloor} -{g}_i  
\\
I(\IS) &=& \{\delta(g_1),\ldots,\delta(g_s)\}
\subset\Z[\Y_\bit,\G_\bit],\nonumber
\end{eqnarray*}
where $\G_\bit=\{G_{ikl}\}$ are Boolean variables
and $\theta_{b_i}(\mathbb G_i)$ is defined in \bref{eq-theta}.
Equation \bref{eq-od} becomes
\begin{eqnarray}\label{eq-od01}
\delta_{\alpha\beta}(o)
&=&
\alpha+\sum_{j=0}^{\beta-1}F_j2^j-o_B  \in\Z[\Y_\bit,  \F_\bit],\\
L_{\alpha\beta} &=& I(\IS) \cup\{\delta_{\alpha\beta}(o)\}
\subset\Z[\Y_\bit, \G_\bit, \F_\bit],\nonumber
\end{eqnarray}
where $\F_\bit=\{F_1,\ldots,F_{\beta-1}\}$ are Boolean variables.

\begin{prop}\label{lm-01in}
We can use {Algorithm \ref{alg-opt}} to solve problem \bref{eq-012}
with probability $\ge 1-\epsilon$ and in time
$\widetilde O(s(m^{2.5}+s^{2.5}\log^{2.5} h)(m+\log h)\kappa^2\log(1/\epsilon) \log u)$
where
$u= 2\sum_{i=1}^m |c_j|+1$, $b=\max_{i=1}^s b_i$,  $h = \max\{u, b\}$,
and $\kappa$ is the maximal condition number of $L_{\alpha\beta}$.
\end{prop}
\begin{proof}
Since $\#\Y_\bit = m$, $\#\G_\bit = s\log b$, and $\#\F_\bit = \log u$,
$L_{\alpha\beta}$ has  $m+s\log b + \log u$ Boolean variables and
total sparseness $s(m+\log b)+ m +1 +\log u$.
Since $u,b \le h$,
$L_{\alpha\beta}$ has  $N_{L_{\alpha\beta}}=O(m+s\log h)$ Boolean variables and
total sparseness $T_{L_{\alpha\beta}}=O(s(m+\log h))$.
By Theorem~\ref{th-m2}, the complexity is
$\widetilde O((m+s\log h)^{2.5}(s(m+\log h))\kappa^2\log(1/\epsilon)\log u$
 $=\widetilde O(s(m^{2.5}+s^{2.5}\log^{2.5} h)(m+\log h)\kappa^2\log(1/\epsilon) \log u)$.
\end{proof}
%

In the rest of this section, we consider the QUBO problem. The QUBO problem is to find an $\Y_\bit = (y_1,\ldots,y_m)^T\in \{0,1\}^m$ that minimizes $\Y_\bit^T Q \Y_\bit $ for an upper-triangular matrix $Q = (Q_{i,j})$ with $Q_{i,j}\in \Z$,
which can be written as the following $(0,1)$-programming problem:
%
%
\begin{eqnarray}\label{eq-dwave}
\min_{\Y\in\{0,1\}^m} o_Q(\Y_\bit) = \Y_\bit^T Q \Y_\bit
\end{eqnarray}
In order to solve this problem, we need to give the lower and upper bounds for the objective function.
Let $Q_{\max} = \max\limits_{i,j} |Q_{i,j}|$.
Since $y_i\in \{0,1\}, 1\le i\le m$, we have $|o(\Y)| \le m^2 Q_{\max} $.

Problem \bref{eq-dwave} can be converted into the standard form
with the new objective function $\widetilde{o}_Q = o_Q + m^2 Q_{\max} $
and $u= 2m^2 Q_{\max} +1$.
Then, we can use {Algorithm \ref{alg-opt}} to solve problem \bref{eq-dwave}.
Let
$$\delta_{\alpha\beta}(o) =
\alpha+\sum_{j=0}^{\beta-1}F_j2^j-\widetilde{o}_Q  \in\C[\Y_\bit, \F_\bit],$$
where $\F_\bit = (F_0,\ldots, F_{\beta-1})\in \F_2^{\beta}$
and $L_{\alpha\beta} = \{\delta_{\alpha\beta}(o)\}$.
%
%
We have
\begin{prop}\label{prop-dwave}
We can use {Algorithm \ref{alg-opt}} to solve problem \bref{eq-dwave}
with probability $\ge 1-\epsilon$ and in time
$\widetilde O(m^{2.5}+\log^{2.5} Q_{\max} )(m^2+\log Q_{\max})\log Q_{\max}\kappa^2\log(1/\epsilon) )$
where $\kappa$ is the maximal condition number of $L_{\alpha\beta}$
and $Q_{\max} = \max\limits_{i,j} |Q_{i,j}|$.
\end{prop}
\begin{proof}
%
Since $o_Q$ is quadratic and the variables are Boolean, we can solve
$L_{\alpha\beta} = \{\delta_{\alpha\beta}(o)\}$ directly with Theorem \ref{th-m2}.
Using the notations in Theorem \ref{th-opt1}, we have $N_{L_{\alpha\beta}}=m+\log(m^2 Q_{\max})
= \widetilde O(m+\log Q_{\max})$,
$T_{L_{\alpha\beta}}=\widetilde O(m^2+\log Q_{\max})$, $u = 2m^2 Q_{\max}+1$.
By Theorem \ref{th-opt1}, the complexity  is
$\widetilde O((m+\log Q_{\max} )^{2.5}(m^2+\log Q_{\max})\kappa^2\log(1/\epsilon) (\log m+\log Q_{\max}))$
$=\widetilde O(m^{2.5}+\log^{2.5} Q_{\max} )(m^2+\log Q_{\max})$ $\log Q_{\max}\kappa^2\log(1/\epsilon) )$.
\end{proof}

\section{Polynomial system with noise}\label{sec-pswn}
In this section, we consider the 
{\em  polynomial systems with noise problem (PSWN)},
which is an optimization problem over finite fields
and has important applications in cryptography \cite{alb1,HL}.

\subsection{Polynomial system with noise}
\label{sec-n1}

\begin{defn}\label{lswn}
Let $p$ be a prime. Given a polynomial system $\FS=\{f_1,\ldots,f_r\}\subset\F_p[\X]$, the PSWN is to find an $\X=(x_1,\ldots,x_n)^\tau\in\F_p^n$ such that
$\FS=\mathbf e$
for the ``smallest" error-vector $\mathbf e=(e_1,\ldots,e_r)^\tau\in\F_p^r$.
\end{defn}
In most cases, the Hamming weight $\|\mathbf e\|_H$ of $\mathbf e$ is used to measure the ``smallness'' and it is assumed that $r\gg n$,
that is, we minimize the number of non-zero components of $\mathbf e$
or satisfy the maximal number of equations of $\FS=0$.
Therefore,  PSWN is also called MAX-POSSO.
We first give the following representation for $\|\mathbf e\|_H$.
\begin{lem}\label{lm-pswn1}
Let $\mathbf e=(e_1,\ldots,e_r)^\tau\in\F_p^n$ and $H_j=e_j^{p-1}$ in $\F_p$. Then $H_j$ is Boolean and $\|\mathbf e\|_H=\sum_{j=1}^{m}H_j$ when the summation is over {$\C$}.
\end{lem}
\begin{proof}
$e_j\in\F_p$ implies $H_j=e_j^{p-1}$ is either $0$ or $1$ in $\F_p$,
and $H_j=1$ if and only if $e_j\ne0$.
Then, $H_j$  is a Boolean variable. Thus, we have $\sum_{j=1}^{m}H_j=\|\mathbf e\|_H$ when the summation is over $\C$.
\end{proof}

Let
\begin{equation}\label{eq-fs0}
E(\FS) =(\FS-\mathbf e)\cup\{H_j-e_j^{p-1}\,|\,j=1,\ldots,r\}
\subset\F_p[\X,\ES,\HS_\bit]
\end{equation}
where  $\HS_\bit=\{H_1,\ldots,H_r\}$ are Boolean variables
and $\ES=\{e_1,\ldots,e_r\}$ are variables over $\F_p$.
By Lemma \ref{lm-pswn1}, PSWN can be formulated as the following
optimization problem over finite fields:
\begin{eqnarray}\label{eq-pswn}
\min_{\X\in\F_p^n} o(\X)=\sum_{j=1}^rH_j&&\hbox{ subject to }
  0\le o(\X)\le r;   E(\FS)=0 \mod p
\end{eqnarray}
which can be solved by Algorithm \ref{alg-opt}.

Due to the special structure of $E(\FS)$, we can achieve better complexities
than that given in Theorem \ref{th-opt1}.
Following \bref{eq-L},  the equation set $L_{\alpha\beta}$  for  PSWN is
\begin{eqnarray}\label{eq-lab}
\delta_{\alpha\beta}(o)
&=&
\alpha+\sum_{j=0}^{\beta-1}F_j2^j- \sum_{j=1}^rH_j \in\C[\HS_\bit,\F_\bit],\cr
L_{\alpha\beta}(\FS)&=&P(Q(E(\FS)))\cup \{\delta_{\alpha\beta}(o)\}\subset\C[\X_\bit,\mathbf \ES_\bit,\HS_\bit,\mathbb F_\bit,\V_\bit,\U_\bit],
\end{eqnarray}
where $\F_\bit=\{F_1,\ldots,F_{\beta-1}\}$ are Boolean variables.
We have
\begin{prop}\label{th-n11}
There is a quantum algorithm to solve PSWN  in time
$\widetilde O(
n^{3.5}T_\FS^{3.5}\log^{8} p\kappa^2\log1/\epsilon)
$
and with probability $\ge 1-\epsilon$, where $T_\FS$ is the total sparseness of $\FS$, and $\kappa$ is the extended condition number of $\FS$.
\end{prop}
\begin{proof}
We first give the complexity of Step 4 of Algorithm \ref{alg-opt},
that is, the complexity to solve
$L_{\alpha\beta}(\FS)$.
Let $\FS_1=\{f_1-e_1,\ldots,f_r-e_r\}$, $\FS_2=\{H_1-e_1^{p-1},\ldots,H_r-e_r^{p-1}\}$, and $\FS_3=\FS_1\cup\FS_2$. Then  $P(Q(\FS_3))=P(Q(\FS_1))\cup P(Q(\FS_2))$.
By Lemma \ref{lm-fsb},  $T_{P(Q(\FS_1))}=O(T_\FS D\log^2 p)$ and
 $N_{P(Q(\FS_1))}=O(T_\FS D\log p)$.
Since each monomial in $\FS_2$ depends on one indeterminate, we have $Q(\FS_2)\subset\F_p[\HS_\bit,\mathbf e,\V]$,
where  $\#\V=O(r\log p),\#Q(\FS_2)=O(r\log p)$, and $T_{Q(\FS_2)}=O(r\log p)$ by the proof for {Lemma \ref{lm-Q}}.
By {Lemma \ref{lm-pf}}, $T_{P(Q(\FS_2))}=O(r\log^3 p)$ and $N_{P(Q(\FS_2))}=O(r\log^2 p)$.
Since $x\in\F_p$ implies $x^p=x$, we can assume $\deg_{x_i}f_j<p$. Thus $D=n+\sum_{i=1}^r\lfloor\log_2\max_j\deg_{x_i}f_j\rfloor\le n+n\lfloor\log_2(p-1)\rfloor= O(n\log p)$, and we have $T_{L_{\alpha\beta}}=O(T_\FS D\log^2 p+r\log^{3} p+r+\log p)=O(T_\FS n\log^{3} p)$ and  $N_{L_{\alpha\beta}}=O(T_\FS D\log p+r\log^2 p+\log r)=O(T_\FS n\log^2 p)$. By {Theorem \ref{th-m2}}, the complexity   to solve $L_{\alpha\beta}$ is
$\widetilde O((T_\FS n\log^2 p)^{2.5}(T_\FS n\log^2 p+T_\FS n\log^{3} p)\kappa^2\log1/\epsilon)
=\widetilde O((
n^{3.5}T_\FS^{3.5}\log^{8} p\kappa^2\log1/\epsilon)$.
The number of loops is at most $\log r$, which is negligible
since $r\le T_\FS$, and the complexity of the algorithm is
that of Step 4. The theorem is proved.
\end{proof}

Similar to Corollary \ref{cor-mqfp}, if $\FS$ is an MQ then the complexity is lower.
\begin{cor}\label{cor-n1}
There is a quantum algorithm to solve the {\em MQ with noise}  in time
$\widetilde O(
(n+r\log p)^{2.5}\\(T_\FS\log p+r\log^2p+n)\log^{3.5}p\kappa^2\log1/\epsilon)
$
with probability $1-\epsilon$.
\end{cor}

\subsection{Linear system with noise}
\label{sec-n2}

When  $\FS$ becomes a linear system, we
obtain the {\em linear system with noise (LSWN)}~\cite{JH}.
Given a matrix $A=(A_{ij})\in\F_p^{r\times n}$ and a vector $\mathbf b=(b_1,\ldots,b_r)^\tau\in\F_p^r$. The LSWN problem is
to find an $\X$ such that $A\X - \mathbf b=\mathbf e$ and the error-vector $\mathbf e\in\F_p^r$ has minimal
Hamming weight $\|\mathbf e\|_H$.

The algorithm given in Section \ref{sec-n1} can be used
to solve the LSWN  and
Proposition \ref{th-n11} becomes the following form.
\begin{prop}\label{th-n2}
There exists a quantum algorithm to solve LSWN with probability $\ge 1-\epsilon$ and in time
$\widetilde O((n + r\log p)^{2.5}(T_A+r\log^{2} p)\log^{3.5} p\kappa^2\log1/\epsilon)$, where $T_A\ge\max\{r,n\}$ is the number of nonzero entries in $A$, and $\kappa$ is the extended condition number of $A\X$.
\end{prop}
\begin{proof}
Similar to the proof of Proposition \ref{th-n11},
we need only consider the complexity of solving $L_{\alpha\beta}(A\X-\mathbf b)$.
Let $f_i=\sum_{j=1}^nA_{ij}x_j-b_i-e_i\in\F_p[\X,\mathbf e]$, $\FS_1=\{f_1,\ldots,f_r\}$, $\FS_2=\{H_1-e_1^{p-1},\ldots,H_r-e_r^{p-1}\}$ , and we have $E(A\X-\mathbf b)=\FS_1\cup\FS_2$. Since $\FS_1$ is a linear system, we have $P(Q(E(A\X-\mathbf b)))=P(\FS_1)\cup P(Q(\FS_2))$.
By Corollary \ref{cor-line},  $T_{P(\FS_1)}=O(T_A\log p)$ and
$N_{P(\FS_1)}=n\log p + \sum_{i=1}^r\log t_i+r\log\log p$, where $t_i$ is the sparseness for the $i$-th row of matrix $A$.
%
By {Lemma \ref{lm-pf}}, $T_{P(Q(\FS_2))}=O(r\log^{3} p)$ and 
$N_{P(Q(\FS_2))}=O(r\log^2 p)$. 
Thus, $T_{L_{\alpha\beta}(A\X-\mathbf b)}=O(T_A\log p+r\log^{3} p)$
 and  $N_{L_{\alpha\beta}(A\X-\mathbf b)}=O(n\log p + \sum_{i=1}^r\log t_i+r\log^2p)$. 
 By {Theorem \ref{th-m2}}, the complexity  to solve $L_{\alpha\beta}(A\X-\mathbf b)$ is
{\small
\begin{eqnarray*}
&&\widetilde O((n\log p + \sum_{i=1}^r\log t_i+r\log^2p)^{2.5}((n\log p + \sum_{i=1}^r\log t_i+r\log^2p)+(T_A\log p+r\log^{3} p))\kappa^2\log1/\epsilon)\\
=&&\widetilde O((n\log p +r\log^2p)^{2.5}(n\log p + T_A\log p+r\log^{3} p)\kappa^2\log1/\epsilon).
\end{eqnarray*}
}
Since $T_A\ge r$ and we can assume $T_A\ge n$ without loss of generality, the complexity is $\widetilde O((n\log p + r\log^2p)^{2.5}(T_A\log p+r\log^{3} p)\kappa^2\log1/\epsilon)=\widetilde O((n + r\log p)^{2.5}(T_A+r\log^{2} p)\log^{3.5} p\kappa^2\log1/\epsilon).$
\end{proof}

%

\section{Short integer solution problem}
\label{sec-sis}
In this section, we consider the {\em short integer solution problem (SIS)},
which is a basic problem in the latticed based cryptosystems~\cite{aj1}.

\subsection{Short integer solution problem}
Consider the SIS problem introduced in \cite{aj1}:
\begin{defn}\label{sis}
Let $A=(A_{ij})\in\F_p^{r\times n}$. The {\em SIS} is to find an $\X\in\F_p^n$ such that $A\X=0\pmod p$ and the Euclidean norm of $\X$ satisfies {$0<\|\X\|_2\le b$}, where $b$ is a given integer.
%
\end{defn}
We first consider the more general {\em SIS} for $\FS=\{f_1,\ldots,f_r\}\subset\F_p[\X]$:
find an $\X$ such that $\FS(\X)=0\pmod p$ and {$0 < \|\X\|_2\le b$}.
Note that SIS is a special case of the optimization problem \bref{eq-op},
where the objective function is a constant
and the problem is to find a feasible solution for the constraints.
Precisely, the SIS can be formualted as the following standard form
 \begin{eqnarray}\label{eq-sis}
\min_{\X\in\F_p^n} o=1 &&\hbox{ subject to }\quad
\FS=0\mod p,\
0 \le \|\X\|_2^2 -1 \le b^2-1.
\end{eqnarray}

From  Remark \ref{rem-xbit}, the representation for $\F_p$ affects
inequality constraints. For the inequality
$0< \|\X\|_2^2\le b^2$,
a better representation for $\F_p$ is $\{-\frac{p-1}{2},\ldots,\frac{p-1}{2}\}$,
instead of $\{0,1,\ldots,p-1\}$.
In this section, we still use $\{0,1,\ldots,p-1\}$ to represent elements in $\F_p$,
but use the following variable expansion instead of \bref{eq-xbit}:
\begin{equation} \label{eq-xbitn} x_i=\overline{\theta}_{p-1}(\X_i)-\frac{p-1}{2}=\sum_{j=0}^{\lfloor\log_2(p-1)-1\rfloor}
X_{i,j}2^j+(p-2^{\lfloor\log_2(p-1)\rfloor})X_{i,{\lfloor\log_2(p-1)\rfloor}}
-\frac{p-1}{2}
\end{equation}
where $\X_i$ are defined in \bref{eq-xbit}.
Then, $x_i$ takes values $-\frac{p-1}{2},\ldots,\frac{p-1}{2}$ when evaluated over $\C$.
The following easy result shows that this representation gives the ``global"
solution to problems involving the Euclidean norm.
\begin{lem}\label{lm-sis1}
For $\check\X\in[-\frac{p-1}{2},\frac{p-1}{2}]^n$ and any vector { $\mathbf v\in(p\Z)^n\backslash\mathbf0$,
$\|\check\X\|_2<\|\check\X+ {\mathbf v}\|_2$}.
\end{lem}
Due to \bref{eq-yt2} and by Lemma \ref{lm-sis1},
the constraint {$0\le \|\X\|_2^2 -1 \le b^2 -1$ can be written as the following MQ in Boolean variables
\begin{eqnarray}\label{eq-ds1}
\overline{\delta}_{b}&=&
\overline{\theta}_{b^2-1}(\G_\bit) -\sum_{i=1}^n(
\overline{\theta}_{p-1}(\X_i)-\frac{p-1}{2})^2 +1 \in\C[\G_\bit,\X_\bit]
\end{eqnarray}
where $\G_\bit = \{G_{k}\,|\,k=0,\ldots,\lfloor\log_2(b^2-1)\rfloor\}$.}
From the above discussion, we have
\begin{lem}\label{lm-sis5}
To solve the SIS, we need only to find a solution of
$\overline{P}(Q(\FS))\cup\{\overline{\delta}_{b}\}\subset\C[\X_\bit,\V_\bit,\U_\bit,$ $\G_\bit]$,
where $\overline{P}(Q(\FS))$ is obtained similar to ${P}(Q(\FS))$,
but using \bref{eq-xbitn} to expand $\X$, $\V$, and $\U$.
\end{lem}

\begin{prop}\label{th-sis1}
There is a quantum algorithm to solve the SIS problem \bref{eq-sis}
with  probability at least $1-\epsilon$
and complexity
$\widetilde O(n^{3.5}T_\FS^{3.5}\log^{3.5} d \log^{4.5} p\kappa^2\log1/\epsilon)$
where $T_\FS$ is the total sparseness of $\FS$,
$d=\max\{2, \log_2(\deg_{x_i}(f_j)),i=1,\ldots,n,j=1,\ldots,r\}$,
and $\kappa$ is the condition number of $P(Q(\FS))\cup\{\overline{\delta}_{b}\}$.
\end{prop}
\begin{proof}
Since $\|\X\|_2^2 \le np^2$,
by Lemma \ref{lm-bb}, {
$\#\G_\bit= \widetilde{O}(\log(b^2-1))\le \widetilde{O}(\log(np^2))=
\widetilde{O}(\log(n)+\log p)$ and
$\#\overline{\delta}_{b}= \widetilde{O}(\log(b^2-1)+n\log^2p)=\widetilde{O}(\log(b)+ n\log^2p)=\widetilde{O}(n\log^2p)$, since $b\le np^2$.}
%
By {Lemma \ref{lm-fsb}},
$\overline{P}(Q(\FS))$ is of sparseness $O(nT_\FS \log d\log^2 p)$ and with $O(nT_\FS\log d \log p)$ indeterminates.
By Lemma \ref{lm-sis5}, we need to solve $\overline{P}(Q(\FS))\cup\{\overline{\delta}_{b}\}$ with Theorem \ref{th-m2}.
Comparing to the total sparseness and number of variables of $\overline{P}(Q(\FS))$,
$\#\overline{\delta}_{b}$ and $\#\G_\bit$ are negligible.
Then, the complexity of solving the SIS is the same as that of
solving $\overline{P}(Q(\FS))$. Then, the theorem follows from Corollary \ref{cor-fsc}.
\end{proof}

For the original SIS, we have
%
\begin{prop}\label{cor-sis2}
There is a quantum algorithm to find an non-trivial $\X\in\Z^n$ for $A\X=0\pmod p$ with $\|\X\|_2\le b$ with probability $1-\epsilon$  and in time $\widetilde O((n\log p+r)^{2.5}(T_A\log p+n\log^2p)\kappa^2\log1/\epsilon)$, where $T_A$ is the
number of nonzero elements of $A$, assuming $T_A\ge n$.
\end{prop}
\begin{proof}
By {Corollary \ref{cor-line}},  $P(Q(A\X))=P(A\X)$ is of total sparseness $O(T_A\log p)$ and
has  $O(n\log p+\sum_{i=1}^r\log t_i+r\log\log p)$ indeterminates, where $t_i$ is the sparseness for the $i$-th line of matrix $A$..
From the proof of Proposition \ref{th-sis1},
$\#\G_\bit=\widetilde{O}(\log n+\log p)$ and
$\#\overline{\delta}_{b}=\widetilde{O}(n\log^2p)$.
Since $T_A\ge n$, we have $T_{L_{\alpha\beta}} = O(T_A\log p+n\log^2p)$
and $N_{L_{\alpha\beta}} = O(n\log p+\sum_{i=1}^r\log t_i+r\log\log p)$.
Comparing to the total sparseness and number of variables of $\overline{P}(Q(\FS))$,
$\#\overline{\delta}_{b}$ is negligible.
By Theorem \ref{th-m2}, the complexity to solve $\overline{P}(Q(\FS))\cup\{ \overline{\delta}_{b}\}$   is $\widetilde O((n\log p+\sum_{i=1}^r\log t_i+r\log\log p)^{2.5}((n\log p+\sum_{i=1}^r\log t_i+r\log\log p)+(T_A\log p+n\log^2p))\kappa^2\log1/\epsilon)=\widetilde O((n\log p+r)^{2.5}(T_A\log p+n\log^2p)\kappa^2\log1/\epsilon)$.
\end{proof}

\subsection{Smallest integer solution problem}
\label{sec-ssis}
We consider the smallest integer solution problem,
which is to find a solution of $\FS=0\pmod p$, which
has the minimal Euclidean norm.
The problem can be formulated as the following standard form
\begin{eqnarray}\label{eq-ssis}
\min_{\X\in\F_p^n} o=\|X\|_2^2-1&&\hbox{ subject to }\quad
\FS=0\mod p \,\, \hbox{ and}\,\, 0 \le o < u
\end{eqnarray}
where $u=n (p-1)^2$.
We can use Algorithm \ref{alg-opt} to solve problem \bref{eq-ssis}.
The parameterized objective function  and $L_{\alpha\beta}(\FS)$ are
\begin{eqnarray}\label{eq-lab1}
\delta_{\alpha\beta}(o)
&=&
\alpha+\sum_{j=0}^{\beta-1}F_j2^j-
\sum_{i=1}^n(\theta(\X_i)-\frac{p-1}{2})^2+1\in\C[\F_\bit,\X_\bit].\\
{L}_{\alpha\beta}&=&\overline{P}(Q(\FS))\cup \{\delta_{\alpha\beta}\}\subset\C[\X_\bit,\V_\bit,\U_\bit,\F_\bit]\nonumber
\end{eqnarray}
where $\overline{P}(Q(\FS))$ is defined in Lemma \ref{lm-sis5}.
We have
\begin{prop}\label{thm-sis1}
There is a quantum algorithm to solve problem \bref{eq-ssis} with probability $\ge1-\epsilon$ and  in time
$\widetilde O(n^{3.5}T_\FS^{3.5}\log^{3.5} d \log^{4.5} p\kappa^2\log1/\epsilon)$.
\end{prop}
\begin{proof}
From the proof of Proposition  \ref{th-sis1},
$N_{L_{\alpha\beta}} =O(nT_\FS\log d \log p)$ and
$T_{L_{\alpha\beta}} = O(nT_\FS \log d\log^2 p)$.
Also, $\log u = O(np^2 ) =O(\log n + \log p)$.
By Theorem \ref{th-opt1}, the complexity is
$\widetilde O(N_{L_{\alpha\beta}}^{2.5}T_{L_{\alpha\beta}}\kappa^2\log1/\epsilon$ $\log u)$
$=\widetilde O(n^{3.5}T_\FS^{3.5}\log^{3.5} d \log^{4.5} p\kappa^2\log1/\epsilon)$.
\end{proof}

If $\FS$ is a linear system $A\X=0$ with $T_A\ge n$, then we can prove the following result
similar to  Propositions \ref{thm-sis1} and \ref{cor-sis2}.
\begin{prop}\label{prop-sis1}
There is a quantum algorithm to find a non-trivial $\X\in\Z^n$ for $A\X=0\pmod p$ with minimal $\|\X\|_2$ with probability $\ge1-\epsilon$  and in time $\widetilde O((n\log p+r)^{2.5}(T_A\log p+n\log^2p)\kappa^2\log1/\epsilon)$.
\end{prop}

\section{Quantum algorithm for SVP and CVP}
\label{sec-svp}

In this section, Algorithm \ref{alg-opt} is used 
to solve the SVP and CVP problems \cite{svp1,svp2}.

A {\em lattice} generated by ${\mathcal B} = \{{\bf b}_1, \ldots, {\bf b}_n\}\subset \R^m$
is the set of $\Z$-linear combinations of ${\bf b}_i$.
${\mathcal B}$ is called a basis of the lattice, if ${\bf b}_1, \ldots, {\bf b}_n$
are linear independent over $\R$.
The {\em SVP problem} can be described as follows:
given a lattice $L$ generated by a basis
${\bf b}_1, \ldots, {\bf b}_n$ in $\R^m$,
find a nonzero ${\bf v}\in L$ such that ${\bf v}$ has the minimal Euclidean norm.
The {\em CVP problem} can be described as follows:
given a vector ${\bf b_0}\in\Z^m$ and a lattice $L$ generated by a basis
${\bf b}_1, \ldots, {\bf b}_n$ in $\R^m$,
find a  ${\bf v}\in L$ such that ${\bf v}-{\bf b_0}$ has the minimal Euclidean norm.
In this paper, we assume that ${\bf b}_1, \ldots, {\bf b}_n$ are in $\Z^m$.
The SVP problem can be written as the following optimization problem.
\begin{eqnarray}\label{eq-svp}
\min_{{\bf v}\in \Z^m, {\bf a}\in \Z^n,} o=\|{\bf v}\|_2^2 &&\hbox{ subject to }\quad
{\bf v} \ne {\bf 0} \,\,\hbox{and}\,\, {\bf v} = \sum_{i=1}^n a_i{\bf b}_i,
\end{eqnarray}
where ${\bf v}=(v_1,\ldots,v_m)$ and ${\bf a}=(a_1,\ldots,a_n)$.
Note that the SVP problem is similar to the SIS problem considered in
Proposition \ref{prop-sis1}, but the solutions are over the integers instead
of finite fields.

Let $$B=[{\bf b}_1,\ldots,{\bf b}_n]\in\Z^{m\times n}$$ be the matrix with columns
${\bf b}_1,\ldots,{\bf b}_n$.
%
In order to reduce problem~\bref{eq-svp} into the standard form~\bref{eq-op},
we need to find upper bounds for $a_i$, $v_i$, and $\|{\bf v}\|_2$.
For a matrix or a vector $A$, let $\|A\|_\infty$ to be the maximum absolute value of the elements in $A$. It is easy to find bounds for $v_i$ and $\|{\bf v}\|_2$.
\begin{lem}\label{lm-svp1}
Let $\overline{{\bf v}}=(\overline{v}_1,\ldots,\overline{v}_m)$ be the shortest vector in $L$.
Then we have $\|\overline{{\bf v}}\|_2 \le \sqrt{m}||B||_\infty$
and $|\overline{v}_i|\le \sqrt{m}||B||_\infty$.
\end{lem}
\begin{proof}
It is clear that $\|\overline{{\bf v}}\|_2 \le \sqrt{m}||B||_\infty $
and $|\overline{v}_i| \le \|\overline{{\bf v}}\|_\infty \le \|\overline{{\bf v}}\|_2 \le \sqrt{m}||B||_\infty $.
\end{proof}

In order to bound $a_i$ in \bref{eq-svp}, we need the concept of
{\em Hermite normal form} (HNF).
A matrix $H=(h_{i,j})\in\Z^{m\times n}$ of rank $n$ is called an (column) HNF if there
exists a strictly increasing map $f$ from $[1,n]$ to $[1,m]$
satisfying: for $j\in[1,n]$, $h_{f(j),j}\ge 1$, $h_{i,j}=0$ if $i > f(j)$
and $h_{f(j),j}> h_{f(j),k}\ge0$ if $k > j$.
It is known that any lattice generated by a basis ${\bf b}_1, \ldots, {\bf b}_n$
is also generated by ${\bf h}_1, \ldots, {\bf h}_n$  if
$H = [{\bf h}_1, \ldots, {\bf h}_n]$ is an HNF of $B = [{\bf b}_1, \ldots, {\bf b}_n]$.
We need the following obvious property of HNF.
\begin{lem}\label{lm-svp2}
Let $H = [{\bf h}_1, \ldots, {\bf h}_n]=[h_{ij}]$ be an HNF,
${\bf v}=(v_1,\ldots,v_m)^\tau$ an element in the lattice
generated by ${\bf h}_1, \ldots, {\bf h}_n$,
and ${\bf v}=c_1 {\bf h}_1+ \cdots+ c_n {\bf h}_n$ for $c_i\in\Z$.
Then $v_{f(n)} = c_n h_{f(n),n}$ and hence $|c_n| \le  ||{\bf v}||_\infty$.
%
\end{lem}
We also need the following result about HNF.
\begin{thm}[\cite{Storjohann2000}]\label{th-hnfc}
Let $B \in \Z^{m\times n}$ with rank $n$ and $H$ be the HNF of $B$.
Then, there exists an $E\in\Z^{m\times m}$ such that $EB = H$,
$||H||_\infty \le (\sqrt{n}||B||_{\infty})^n$ and $||E||_{\infty} \le (\sqrt{n}||B||_{\infty})^n$.
Furthermore, the bit complexity to compute $H$ from $B$ is
$\widetilde{O}(mn^{\theta}||B||_{\infty})$, where $\theta$ is the matrix multiplication constant.
\end{thm}

We now give a bound for $a_i$ in \bref{eq-svp}.
\begin{lem}\label{lm-svp3}
Let $\overline{{\bf v}} $ be the shortest vector in $L$.
Then there exist $a_1,\ldots, a_n \in \Z$
such that $\overline{{\bf v}} = \sum_{i=1}^n a_i{\bf b}_i$ and $|a_i| \le b_B$,
where $b_B = n \sqrt{m}||B||_\infty ((\sqrt{n}||B||_{\infty})^n+1)^{n+1}$.
\end{lem}
\begin{proof}
Let $H=EB$ be the HNF of $B$ and ${\bf h}_i$ the $i$-th column of $H$, $1\le i\le n$.
Then there exist $c_i, 1\le i\le n$, such that $\overline{{\bf v}} = \sum_{i=1}^n c_i {\bf h}_i$.
Denote $e_1=\sqrt{m}||B||_\infty$ and $e_2=(\sqrt{n}||B||_{\infty})^n+1$.
We claim that $|c_i| \le e_1(e_2+1)^n$.
Let $\overline{\bf v}_{i}=c_1 {\bf h}_{1}+\cdots+c_{i} {\bf h}_{i}$
for $i=1,\ldots,n$.
We prove the claim by proving $||\overline{\bf v}_{i}||_\infty \le e_1(e_2+1)^{n-i}$
and $|c_i|\le e_1(e_2+1)^{n-i}$ by induction for $i=n,n-1,\ldots,i$.
By Lemma \ref{lm-svp2}, the second inequality comes from the first one: $|c_i|\le||\overline{\bf v}_{i}||_\infty \le e_1(e_2+1)^{n-i}$,
since $[{\bf h}_1,\ldots,{\bf h}_i]$ is also an HNF.
By Lemma \ref{lm-svp2}, $|c_n| \le e_1$ and the case of $i=n$ is true.
Suppose the claim is true for $i=n,\ldots,j+1$.
By Lemma \ref{lm-svp1}, we  have  $||{\bf h}_{j+1}||_\infty\le e_2$.
Since $\overline{\bf v}_{j}= \overline{\bf v}_{j+1} - c_{j+1} {\bf h}_{j+1}$,
we have $\|\overline{\bf v}_{j}\|_\infty \le
||\overline{\bf v}_{j+1}||_\infty + |c_{j+1}| ||{\bf h}_{j+1}||_\infty
\le e_1(e_2+1)^{n-j-1} + e_1(e_2+1)^{n-j-1} e_2 = e_1(e_2+1)^{n-j}$.
The claim is proved.

We have $\overline{{\bf v}} = \sum_{i=1}^n c_i {\bf h}_i = (c_1,\ldots, c_n) H = (c_1,\ldots, c_n)EB =a_1{\bf b}_1 +\cdots+a_n{\bf b}_n$. Then $(a_1,\dots, a_n) = (c_1,\ldots, c_n)E$,
and hence $|a_i| \le n\max_i |c_i|\cdot ||E||_\infty \le ne_1(e_2+1)^n e_2
\le ne_1(e_2+1)^{n+1}$ by Theorem~\ref{th-hnfc}.
\end{proof}

By the above lemma, we can rewrite the SVP as the standard form  \bref{eq-op}:
\begin{eqnarray}\label{eq-svpnew}
\min_{{\bf v}\in \Z^m, {\bf a}\in\Z^n} && o=\|{\bf v}\|_2^2-1   \quad \hbox{ subject to } 0 \le o < m||B||_\infty^2\cr
&&
 {\bf v} = a_1{\bf b}_1+\cdots+a_n{\bf b}_n, \cr
 && 0\le a_i + b_B \le 2 b_B, \quad 1\le i\le n, \cr
 && 0 \le v_i + \sqrt{m}||B||_\infty \le 2\sqrt{m}||B||_\infty , \quad 1\le i\le m\nonumber
\end{eqnarray}
where $b_B$ is given in Lemma \ref{lm-svp3},
and the arguments are ${\bf v} = (v_1,\ldots, v_m)$
and ${\bf a} = (a_1,\ldots, a_n)$.

Note that, the above problem is already an MQ,
so we just need to change the variables to Boolean variables as follows by using Lemma \ref{lm-bb}.
\begin{eqnarray}\label{eq-svpbool}
 &&a_i = \theta_{2b_B}(A_{i,0},\ldots,A_{i,\lfloor\log_2(2b_B)\rfloor}) - b_B, \quad 1\le i\le n, \cr
&& v_i =\theta_{2\sqrt{m}||B||_\infty}(V_{i,0},\ldots,V_{i,\lfloor\log_2(2\sqrt{m}||B||_\infty)\rfloor}) - \sqrt{m}||B||_\infty, \quad 1\le i\le m
\end{eqnarray}
where  and $A_{i,j}, V_{i,j}$ are Boolean variables.

Then we can use {Algorithm \ref{alg-opt}} to solve the  problem for $u = m||B||_\infty^2$ in the input.
For the objective function $o$, we denote by $\overline{o}$ the
Boolean function obtain from $o$ by replacing the $v_i$ by the above equation~\bref{eq-svpbool}.
Let
\begin{eqnarray}\label{eq-ooo}
\delta_{\alpha\beta}(\overline{o})
&=&
\alpha+\sum_{j=0}^{\beta-1}C_j2^j-\overline{o} +1,\\
L_{\alpha\beta} &=& \CS\cup\{\delta_{\alpha\beta}(\overline{o})\},\nonumber
\end{eqnarray}
where $\CS$ is obtained from ${\bf v} = \sum_{i=1}^n a_i{\bf b}_i$ by replacing the $a_i, v_i$ by the equation~\bref{eq-svpbool}.
%
%
\begin{prop}\label{prop-svp1}
There exists a quantum algorithm to solve the SVP  with probability $\ge 1-\epsilon$ and in time
$\widetilde O(m(n^{7.5} + m^{2.5})(n^3+\log h)\log^{4.5}  h \kappa^2\log\frac{1}{\epsilon})$,
%
where $h=||B||_\infty$ and $\kappa$ is the maximal condition number of $L_{\alpha\beta}$.
 \end{prop}
\begin{proof}
The numbers of $\{A_{i,j}\}$, $\{V_{i,j}\}$, and $\{C_{j}\}$
in \bref{eq-svpbool} and \bref{eq-ooo}
are $n\log_2 (2b_B)$, $m \log_2 (\sqrt{m}||B||_\infty)$ and $\log_2 ({m}||B||_\infty^2)$,
respectively.
So,
$N_{L_{\alpha\beta}} =
n\log_2 (2b_B) + m \log_2 (\sqrt{m}||B||_\infty) + \log_2 ({m}||B||_\infty^2)
=\widetilde O(n^3\log h + m\log h)$.
The total sparseness of $\CS$ is $O(m (\log(\sqrt m\|B\|_\infty) + n\log (2b_B)))$
and the total sparseness of $\delta_{\alpha\beta}(\overline{o})$
is $\log_2 ({m}||B||_\infty^2)+ m\log^2(\sqrt{m}||B||_\infty)$.
So the
total sparseness of $L_{\alpha\beta}$ is
$T_{L_{\alpha\beta}}= m (n\log (2b_B)+\log(\sqrt m\|B\|_\infty))+ m\log^2(\sqrt{m}||B||_\infty)+\log_2 ({m}||B||_\infty^2)
=\widetilde O(mn^3\log h + m\log^2 h)$.
By Theorem~\ref{th-m2}, the SVP can be solved in time
$\widetilde O(N_{L_{\alpha\beta}}^{2.5}T_{L_{\alpha\beta}}\kappa^2\log\frac{1}{\epsilon} \log ({m}||B||_\infty^2))=
\widetilde O(m(n^{7.5} + m^{2.5})(n^3+\log h)\log^{4.5}  h \kappa^2\log\frac{1}{\epsilon})$.
\end{proof}

The CVP can be solved similar to SVP, where the only difference is
that the objective function is $o=\|{\bf v} - {\bf b_0}\|_2^2-1 < m(||B||_\infty+||{\bf b_0}||_\infty)^2$. Similar to Proposition \ref{prop-svp1}, we have
\begin{prop}
There exists a quantum algorithm to solve the CVP  with probability $\ge 1-\epsilon$ and in time
$\widetilde O(m(n^3\log h + m\log h + \log h_0)^{2.5}(n^3\log h+\log^2(h+h_0))  \kappa^2\log\frac{1}{\epsilon})$,
%
where $h=||B||_\infty$, $h_0= ||{\bf b}_0||_\infty$,  and $\kappa$ is the condition number of the problem.
 \end{prop}
%
%
%

\section{Quantum algorithm to recover the private key for NTRU}
In this section, we will give a quantum algorithm
to recover the private key of NTRU from its known public key.

The {\rm NTRU} cryptosystem depends on three integer parameters $(N,p,q)$ and two sets $\mathcal{L}_{f},\mathcal{L}_{g}$
of polynomials in $\Z[X]$ with degree $N-1$. Note that $p$ and $q$ need not to be prime, but we will assume that $\gcd(p,q)=1$, and $q$ is always  considerably larger than $p$.
Denote $\Z_k$  to be the ring $\Z/(k) = \{0,1\ldots,k-1\}$  for any $k\in\Z_{>0}$.
We work in the ring $R=\Z[X]/(X^N-1)$.
An element $F\in R$ will be written as a polynomial or a vector,
\begin{equation}
F=\sum_{i=0}^{N-1}F_ix^i=(F_0,F_1,\ldots,F_{N-1})^\tau.
\end{equation}
Given two positive integers $d_f$ and $d_g$, we set
\begin{eqnarray}
\mathcal L_f&=&\{f\in R\,|\,f\text{ has }d_f \text{ coefficients }1,\  d_f-1\text{ coefficients } -1\text{, and the rest }0\},\\
\mathcal L_g&=&\{g\in R\,|\,g\text{ has }d_g \text{ coefficients }1,\  d_g\text{ coefficients } -1\text{, and the rest }0\}.
\end{eqnarray}
Let $f\in L_f$ be invertible both $\pmod p$ and $\pmod q$.
The private key for {\rm NTRU} is $f$ and the public key is $h=gf^{-1}\pmod q$ for some $g\in \mathcal L_g$.
A set of parameters could be
$(N,p,q) = (107, 3, 64)$, $d_f=15$, and $d_g = 12$ \cite{NTRU}.

We need to find $f$ from $h$.
We will reduce this problem to an equation solving problem over the finite rings
$\Z_p$ and $\Z_q$.
Set $f=(f_0,\ldots,f_{N-1})^\tau$, $g=(g_0,\ldots,g_{N-1})^\tau$, $f^{-1}\pmod p=\mathbf p=(p_0,\ldots,p_{N-1})^\tau$, $f^{-1}\pmod q=\mathbf q=(q_0,\ldots,q_{N-1})^\tau$, and $h=(h_0,\ldots,h_{N-1})^\tau$.
Thus, we have the following equations:
\begin{eqnarray}
f\in\mathcal L_f&\Longleftrightarrow& 
  2d_f = \sum_{i=0}^{N-1} f_i^2 +1, 
  \sum_{i=0}^{N-1}f_i=1\text{ and each }f_i^3-f_i=0;\\
g\in\mathcal L_g&\Longleftrightarrow&
  2d_g= \sum_{i=0}^{N-1} g_i^2,
  \sum_{i=0}^{N-1}g_i=0\text{ and each }g_i^3-g_i=0;\\
h=gf^{-1}\pmod q&\Longleftrightarrow&\sum_{j+k=i,i+N}h_jf_k\equiv g_i\pmod q \text{ for } i=0,\ldots,N-1;\\
f^{-1}\pmod q\text{ exists}&\Longleftrightarrow&\sum_{j+k=i,i+N}q_jf_k\equiv\delta_{0i}\pmod q \text{ for } i=0,\ldots,N-1;\\
f^{-1}\pmod p\text{ exists}&\Longleftrightarrow&\sum_{j+k=i,i+N}p_jf_k\equiv\delta_{0i}\pmod p\text{ for } i=0,\ldots,N-1,
\end{eqnarray}
where $\delta_{0i}=1$ for $i=0$ and $\delta_{0i}=0$ for $i\ne0$.
Let $\X=\{f_i,g_i,h_i,p_i,q_i\,|\,i=0,\ldots,N-1\}$, and
\begin{eqnarray}
\FS_1&=&\{2d_f = \sum_{i=0}^{N-1} f_i^2 +1, 2d_g= \sum_{i=0}^{N-1} g_i^2,\cr
&&\sum_{i=0}^{N-1}f_i-1,\sum_{i=0}^{N-1}g_i,f_i^3-f_i, g_i^3-g_i,i=0,\ldots,N-1\}\subset\C[f,g],\\
\FS_2&=&\{\sum_{j+k=i,i+N}h_jf_k-g_i,\sum_{j+k=i,i+N}q_jf_k-\delta_{0i}\,|\,
i=0,\ldots,N-1\}\subset\Z_q[f,g,h,\mathbf q],\\
\FS_3&=&\{\sum_{j+k=i,i+N}p_jf_k-\delta_{0i}\,|\,i=0,\ldots,N-1\}\subset\Z_p[f,g,h,\mathbf p].
\end{eqnarray}

Note that $\FS_1,\FS_2,\FS_3$ are over $\C$, $\Z_q$, $\Z_p$, respectively.
We can modify the method given in Section \ref{ss-fp} to solve the equation system
$\FS_1=\FS_2=\FS_3=0$.

We first give a simpler treatment for $\FS_1$.
Let $\Z_\bit=\{F_{i1},F_{i2},G_{i1},G_{i2},i=0,\ldots,N-1\}$ be Boolean variables,
$f_i=F_{i1}+F_{i2}-1$ and $g_i=G_{i1}+G_{i2}-1$.
Then, the constraints $f_i^3=f_i$ and $g_i^3=g_i$ are automatically satisfied.
When $F_{i1}=0,F_{i2}=1$ and $F_{i1}=1,F_{i2}=0$, we both have
$f_i=0$. To avoid this redundance, we add an extra equations $F_{i1}F_{i2}-F_{i2}$.
We have $2d_f=\sum_{i=0}^{N-1}f_i^2+1=\sum_{i=0}^{N-1}(F_{i1}+F_{i2}-1)^2+1=
\sum_{i=0}^{N-1}(F_{i1}-F_{i2})+N+1\pmod{(F_{i1}^2-F_{i1},
F_{i2}^2-f_{i2},F_{i1}F_{i2}-F_{i2})}$.
Similarly, $d_g=\sum_{i=0}^{N-1}(G_{i1}-G_{i2})+N \pmod{(G_{i1}^2-G_{i1},
G_{i2}^2-G_{i2},G_{i1}G_{i2}-G_{i2})}$.
Then, $\FS_1$ is equivalent to
\begin{eqnarray}
\FS_{11}&=&\{
\sum_{i=0}^{N-1}(F_{i1}-F_{i2})+N+1 -2d_f,
\sum_{i=0}^{N-1}(G_{i1}-G_{i2})+N-2d_g,\cr
&&\sum_{i=0}^{N-1} (F_{i1}+F_{i2}-1)-1,
\sum_{i=0}^{N-1}(G_{i1}+G_{i2}-1),\cr
&& F_{i1}F_{i2}-F_{i2},G_{i1}G_{i2}-G_{i2},
i=0,\ldots,N-1,\subset\C[\F_\bit]\},
\end{eqnarray}
where $\F_\bit=\{F_{ij},G_{ij}\,|\,i=0,\ldots,N-1;\,j=1,2\}$.
%
%

We can compute  $B(\FS_2)\subset\Z_p[\X_\bit]$ and $B(\FS_3)\subset\Z_p[\X_\bit]$ defined in \bref{eq-s1}
by setting   $q_i=\theta_{q-1}(Q_{i0},\ldots,Q_{i\lfloor\log_2(q-1)\rfloor})$ and $p_i=\theta_{p-1}(P_{i0},\ldots,P_{i\lfloor\log_2(p-1)\rfloor})$, where
\begin{eqnarray*}
\X_\bit&=&\{F_{i1},F_{i2},G_{i1},G_{i2}\,|\,i=0,\ldots,N-1\}\cup\cr
&&\{P_{ij}\,|\,i=0,\ldots,N-1,j=0,\ldots,\lfloor\log_2(p-1)\rfloor\}\cup\cr
&&\{Q_{ij}\,|\,i=0,\ldots,N-1,j=0,\ldots,\lfloor\log_2(q-1)\rfloor\}
\end{eqnarray*}
Note that $\FS_2$ and $\FS_3$ are already MQ, we can compute $P(\FS_2)$ and $P(\FS_3)$ as in \bref{eq-P}.
Therefore, we can use algorithm {\bf QBoolSol} to find a Boolean solution
$\check{\X}$ for
$$\FS_{\rm NTRU}=\FS_{11}\cup P(\FS_2)\cup P(\FS_3)\subset\C[\X_\bit].$$
Finally set $\check f_i=\check F_{i1}+\check F_{i2}-1$, and we have a possible private key $\check f=(\check{f}_0,\ldots,\check{f}_{N-1})$.
%

\begin{prop}\label{p-nt1}
There is a quantum algorithm to obtain a   private key $f$ from the public key $h$ in time $\widetilde O(N^{4.5}\log^{4.5} q\kappa^2\log1/\epsilon)$ with probability $\ge 1-\epsilon$, where $\kappa$ is the condition number for $\FS_{\rm NTRU}$.
\end{prop}
\begin{proof}
Only the complexity need to be considered. $T_{\FS_2}=2N^2+N+1$, $T_{\FS_3}=N^2+1$, $T_{\FS_{11}}=O(N)$. By {Lemma \ref{lm-pf}} and   {Corollary \ref{cor-line}}, $T_{P(\FS_2)}=O(N^2\log^2 q)$, $T_{P(\FS_3)}=O(N^2\log^2 p)$, and then we have $T_{\FS_{\rm NTRU}}=O(N^2(\log^2 q+\log^2 p))=O(N^2\log^2 q)$ and
$ N_{\FS_{\rm NTRU}}=O(N\log q+N\log N+N\log\log q)=\widetilde O(N\log q)$ by Lemma \ref{lm-pf},
where we can ignore $p$ considering $p\ll q$.
By {Theorem \ref{th-m2}}, we can obtain a possible private key $f$ in time $\widetilde O(N^{4.5}\log^{4.5} q\kappa^2\log1/\epsilon)$.
\end{proof}
%

In the design of NTRU, it is assumed that the size of $f$ and $g$
are small. We can use the methods given in Section \ref{sec-n1}
to find $f$ and $g$ which have the smallest $d_f+d_g$.
%
%
\begin{prop}\label{pr-ntru1}
There is a quantum algorithm to obtain a private key $f$ from the public key $h$ such that $d_f+d_g$ is minimal in time $\widetilde O(N^{4.5}\log^{4.5} q\kappa^2\log1/\epsilon)$ with probability $\ge 1-\epsilon$, where $\kappa$ is the extended condition number for $\FS_{\rm NTRU}$.
\end{prop}
\begin{proof}
Remove $\sum_{i=0}^{N-1}(F_{i1}-F_{i2})+N+1 -2d_f$ and 
$\sum_{i=0}^{N-1}(G_{i1}-G_{i2})+N-2d_g$ from $\FS_{1}$
and still denote $\FS_{\rm NTRU}=\FS_{1}\cup \FS_{2}\cup \FS_{3}$.
We can use the objective function $o = (2d_f-1-N) + 2d_g-N -1=
\sum_{i=0}^{N-1}(F_{i1}-F_{i2}+G_{i1}-F_{i2})-1$
 which satisfies $0\le o< 4N$.
Following \bref{eq-od},
we have
$\delta_{\alpha\beta}=\alpha+\sum_{j=0}^{\beta-1}E_j2^j -o$ and
$L_{\alpha\beta}=\FS_{\rm NTRU}\cup\{\delta_{\alpha\beta}\}\subset\C[\X_\bit,\E_\bit]$.
Then we can use {Algorithm \ref{alg-opt}} to find a private key $f$
which minimizes $o$.
By the proof of {Proposition \ref{p-nt1}},
we have  $T_{\FS_{\rm NTRU}}=O(N^2\log^2 q)$ and
$N_{\FS_{\rm NTRU}}=O(N\log q)$.
Then,  $T_{L_{\alpha\beta}}=\widetilde{O}(N^2\log^2 q)$
and $N_{L_{\alpha\beta}}=O(N\log q)$.
%
%
By Theorem \ref{th-opt1}, the complexity is
$\widetilde O(N^{4.5}\log^{4.5} q\kappa^2\log1/\epsilon\log N)$
$=\widetilde O(N^{4.5}\log^{4.5} q\kappa^2\log1/\epsilon)$.
\end{proof}

For the parameters recommended in \cite{NTRU},
$(N,p,q)=(107,3,64)$,
$(N,p,q)=(167,3,128)$,
$(N,p,q)=(503,3,256)$,
and $\epsilon=1\%$,
the complexities in Proposition \ref{pr-ntru1} is given in the following table.
\begin{table}[ht]\centering
\begin{tabular}{|c|c|c|c|}\hline
${N}$&${q}$&$p$& Complexity \\ \hline
107&64&3&$2^{45}\kappa^2$ \\
167&128&3&$2^{49}\kappa^2$ \\
503&256&3&$2^{57}\kappa^2$ \\\hline
\end{tabular}
\caption{Complexities of the quantum algebraic attack on NTRU }
\label{tab-0}
\end{table}

In Table \ref{tab-0},  $\kappa$ is the condition number of the corresponding equation systems. 
From the table, this main part of the complexity is relatively low comparing to its desired security $3^N$ if $\kappa$ is small, which implies that the NTRU is safe only if its condition number is large.

\section{Conclusion}
In this paper, we give quantum algorithms for two basic
computational problems: polynomial system solving over a finite field and
the optimization problem where the arguments  either take values from
a finite field or are bounded integers.
The complexities of these quantum algorithms are polynomial
in the input size, the maximal degree of the inequality constraints,
and $\kappa$ which is the condition number of certain matrices
derived from the problem.
So, we achieve exponential speedup for these problems
when the condition number is small.

The optimization problem considered in this paper covers
many NP-hard problems as special cases.
In particular, the proposed algorithms are used to give quantum algorithms
for several fundamental computational problems in cryptography,
including the polynomial system with noise,
the short solution problem,
the shortest vector problem,
and the NTRU cryptosystem.
The complexity for all of these problems is polynomial
in the input size and their condition numbers,
which means that these problems are difficult
to solve by a quantum computer if and only if
their condition numbers are large.
As a consequence, the NTRU cryptosystem as well as
the candidates recently proposed for
post-quantum standard of public key cryptosystems \cite{alb2}
are safe against quantum computer attacks only if the condition number of
its equation system is large.

The main idea of the algorithm is to convert the equality and
inequality constraints of the optimization problem into
polynomial equations in Boolean variables
and then convert the finding of the minimal value
of the objective function into several problems
of finding the Boolean solutions for polynomial systems
over $\C$, that is B-POSSO. Then the quantum algorithm from
\cite{qabes} is used to find Boolean solutions for these polynomial systems.
%
%

As we just mentioned that the optimization problem
is reduced into the B-POSSO problem.
It is interesting to give a description for all the problems
that can be efficiently reduced to B-POSSO.
It is also interesting to see whether it is possible to
combine the reduction methods introduced in this paper
with traditional algorithms for polynomial system solving such
as the Gr\"obenr basis method \cite{fg1}
and the characteristic set method \cite{cs-ff}
to obtain better traditional algorithms for polynomial
system solving and optimization over finite fields.
Finally, in order to know the exact complexity of the algorithm
proposed in this paper, we need to know the condition number,
which is a main future problem for study.

\end{document}